\documentclass{article}
 
\usepackage{arxiv}

\usepackage[utf8]{inputenc}
\usepackage{amsmath,amssymb,amsfonts,amsthm}
\usepackage{ifthen}
\usepackage{comment}

\pdfoutput=1

\usepackage{url}            
\usepackage{booktabs}       
\usepackage{nicefrac}       
\usepackage{microtype}      
\usepackage{xcolor}         
\usepackage{enumerate}
\usepackage{pgf,pgfplots,graphicx,subcaption,siunitx}     
\usepackage[export]{adjustbox}
\usepackage[normalem]{ulem}
\usepackage{multirow, multicol}
\usepackage{doi}

\newtheorem{definition}{Definition}
\newtheorem{theorem}{Theorem}
\newtheorem{lemma}{Lemma}

\newtheorem{claim}{Claim}
\newtheorem{remark}{Remark}

\usepackage{tikz}
    \usetikzlibrary{matrix,positioning}
    \usetikzlibrary{decorations.pathreplacing}
    \usetikzlibrary{through,calc,intersections}
    \usetikzlibrary{shapes,shapes.geometric}

\newenvironment{proofof}[1]
      {\smallskip\noindent{\bf #1.}}
      {\hfill$\Box$\medskip}
      
\DeclareMathOperator*{\argmax}{arg\,max}

\title{Forward Looking Best-Response Multiplicative Weights Update Methods for Bilinear Zero-sum Games}

\date{} 					

\author{   
  Michail Fasoulakis \\
  Institute of Computer Science, \\
  Foundation for Research and Technology-Hellas \\ (FORTH)\\
  \texttt{mfasoul@ics.forth.gr} \\
  \And
  Evangelos Markakis \\
  Department of Informatics,\\ Athens University of Economics and Business,\\
  \texttt{markakis@gmail.com} \\
  \AND
  Yannis Pantazis \\
  Institute of Applied and Computational Mathematics, \\
  Foundation for Research and Technology-Hellas \\ (FORTH) \\
  \texttt{pantazis@iacm.forth.gr} \\
  \And
  Constantinos Varsos \\
  Institute of Computer Science, \\
  Foundation for Research and Technology-Hellas \\ (FORTH) \\
  \texttt{varsosk@ics.forth.gr} \\
}

\hypersetup{
pdftitle={Forward Looking Best-Response Multiplicative Weights Update Methods for Bilinear Zero-sum Games},
pdfsubject={cs.GT,cs.LG},
pdfauthor={Michail~Fasoulakis, Evangelos~Markakis, Yannis~Pantazis, Constantinos~Varsos},
pdfkeywords={Multiplicative Weights Update, Zero-sum games, Last Iterate Convergence},
}

\begin{document}
\maketitle

\begin{abstract}
Our work focuses on extra gradient learning algorithms for finding Nash equilibria in bilinear zero-sum games. The proposed method, which can be formally considered as a variant of Optimistic Mirror Descent \cite{DBLP:conf/iclr/MertikopoulosLZ19}, uses a large learning rate for the intermediate gradient step which essentially leads to computing (approximate) best response strategies against the profile of the previous iteration. Although counter-intuitive at first sight due to the irrationally large, for an iterative algorithm, intermediate learning step, we prove that the method guarantees last-iterate convergence to an equilibrium.
Particularly, we show that the algorithm reaches first an $\eta^{1/\rho}$-approximate Nash equilibrium, with $\rho > 1$, by decreasing the Kullback-Leibler divergence of each iterate by at least $\Omega(\eta^{1+\frac{1}{\rho}})$, for sufficiently small learning rate, $\eta$, until the method becomes a contracting map, and converges to the exact equilibrium.
Furthermore, we perform experimental comparisons with the optimistic variant of the multiplicative weights update method, by \cite{Daskalakis2019LastIterateCZ} and show that our algorithm has significant practical potential since it offers substantial gains in terms of accelerated convergence. 
\end{abstract}


\section{Introduction}

Our work focuses on the problem of designing learning algorithms for finding Nash equilibria in zero-sum games. 
Zero-sum games form a fundamental class of bimatrix games, where the two players need to solve a max-min and a min-max optimization problem respectively, with a bilinear objective function. 
It is well known by von Neumann's \emph{minmax theorem}, that these two problems have the same optimum. 
Apart from their role in the development of  game theory, zero-sum games also have a prominent role in optimization, as any linear program can be recast to solving such a game \cite{Adler13}. More recently, there has also been a renewed interest in the learning theory community for zero-sum games, given their applications on boosting and reinforcement learning (see \cite{DBLP:conf/icml/DaiS0XHLCS18}), and their relevance in formulating GANs in deep learning (as they capture the interaction between the Generator and the Discriminator, see \cite{GPMXWOCB14}).

Although one can solve a zero-sum game by centralized linear programming algorithms, the application areas above highlight the importance of developing fast, iterative learning algorithms.
Several approaches have been proposed throughout the past decades starting with \emph{fictitious play} \cite{Rob51}.
Recently, some of the more standard methodologies include the family of no-regret algorithms as well as several classes of first-order methods.  
To mention a few examples, the important class of \emph{Multiplicative Weights Update} (MWU) algorithms \cite{LW94,Freund1999AdaptiveGP}, together with Gradient Descent, Mirror Descent, and Extra Gradient methods (for a survey see \cite{Bubeck15}), all fall within the above approaches.

In this work, we are interested in methods that exhibit {\it last-iterate} convergence, a property most desirable from an application point of view, meaning that the strategy profile $(x^t, y^t)$, reached at iteration $t$ of an iterative algorithm, converges to the actual equilibrium as $t\rightarrow \infty$. Unfortunately, many of the methods mentioned above do not satisfy this. No-regret algorithms, like the MWU method, are known to converge only in an average sense, resulting in an $\varepsilon$-Nash equilibrium in expectation (see \cite{DBLP:journals/toc/AroraHK12}) for $\varepsilon>0$. In fact, it was shown in \cite{DBLP:conf/sigecom/BaileyP18} that several MWU variants do not satisfy last-iterate convergence. Similarly, the same can be shown for many descent-based methods (see e.g., \cite{DBLP:conf/iclr/MertikopoulosLZ19}).

Driven by these negative results, recent works have focused on certain {\it optimistic} variations of well known optimization methods. In particular, \cite{Daskalakis2019LastIterateCZ} studied a variant of MWU, referred to as the \emph{Optimistic Multiplicative Weights Update} method (OMWU), where an extra negative momentum term is added to correct the dynamics behavior. Their main result is that for zero-sum games with a unique Nash equilibrium, OMWU exhibits last-iterate convergence. Even further, \cite{DBLP:conf/iclr/MertikopoulosLZ19} considered an extra gradient method, named \emph{Optimistic Mirror Descent} (OMD), where last-iterate convergence for a more general class of min-max optimization problems is  established.
These positive results have generated more interest on the behavior and limitations of such approaches, which is not yet fully understood. Namely, 
they give rise to further questions, 
such as: (i) can we prove last-iterate convergence for other related dynamics, and (ii) can we establish faster convergence rates?
These questions are the main focus of our work. 


\subsection{Our Contribution}
\label{sec:contribution}

We introduce a simple yet substantially different variant of Optimistic Mirror Descent method with entropy regularization \cite{DBLP:conf/iclr/MertikopoulosLZ19}, for the case of zero-sum games. OMD is an extra gradient method, i.e., it contains an intermediate gradient step before the final update step, and each iteration is characterized by its learning rate parameter, which is the same for both steps (and often the same across all iterations). Our tweak
is that the intermediate step uses a different learning rate parameter from the update step in each iteration. In fact, we set this to be sufficiently large, which yields a {\it game-theoretic interpretation}, namely that we compute (approximate) best response strategies against the profile of the previous iteration, as a look ahead move.    
Then, during the final update step, we apply multiplicative weights updates by rewarding more the pure strategies that perform better against the best responses that we found in the intermediate step. 
Consequently, we refer to this OMD variant as Forward-Looking Best-Response - Multiplicative Weights Update (FLBR-MWU) method. 

At first sight, this may look counter-intuitive, since learning rates are usually kept small in classic MWU algorithms and, more generally, in any kind of iterative gradient-type optimization algorithms (apart from the notable exception of \cite{DBLP:conf/nips/BaileyP19}). However, our theoretical and experimental study reveal the following promising findings: 
\begin{itemize}
    \item In Section \ref{sec:theory}, we investigate theoretically the convergence properties of FLBR-MWU. If $\eta$ is the standard learning rate parameter used in the update step, and $\xi$ is the corresponding parameter in the intermediate step, then FLBR-MWU exhibits last-iterate convergence for games with a unique equilibrium, when $\xi$ is sufficiently large and $\eta\xi < 1$. 
    Our proof employs a similar methodology to \cite{Daskalakis2019LastIterateCZ}, adapting convergence tools from the field of dynamical systems.
    Our method also appears to attain faster convergence, quantified in terms of $\eta$, compared to OMD and OMWU. In particular, we prove that the decrease in the divergence from the equilibrium is at least $\Omega(\eta^{1+1/\rho})$ per iteration, for any $\rho>1$, until we reach an approximate $O(\eta^{1/\rho})$-equilibrium, by which time, our rule becomes a contraction map (see also Figure \ref{motiv:ex:fig}). This improves on the $\Omega(\eta^3)$ bound established for OMWU in \cite{Daskalakis2019LastIterateCZ}. Although our bounds do not translate into bounds with respect to time, we suspect a linear convergence rate is highly likely (supported also by our experiments). This has been recently established for OMWU in \cite{Wei2021LinearLC}, and is left as an open problem for FLBR-MWU.
    \item In Section \ref{sec:exp}, we perform numerical experiments, using randomly generated data, comparing FLBR-MWU with OMWU\footnote{We note that for the case of zero-sum games, it has been shown in \cite{Wei2021LinearLC} that OMWU can be seen as a variant of OMD with entropy regularization.}. Our experiments reveal that in practice our method achieves indeed a much faster convergence rate, showing an average speedup by a factor of 10 for small size games and up to hundreds, or even higher, for larger games compared to OMWU.
\end{itemize}


\subsection{A Revealing Example}

\begin{figure*}
\centering
    \includegraphics[scale=.25]{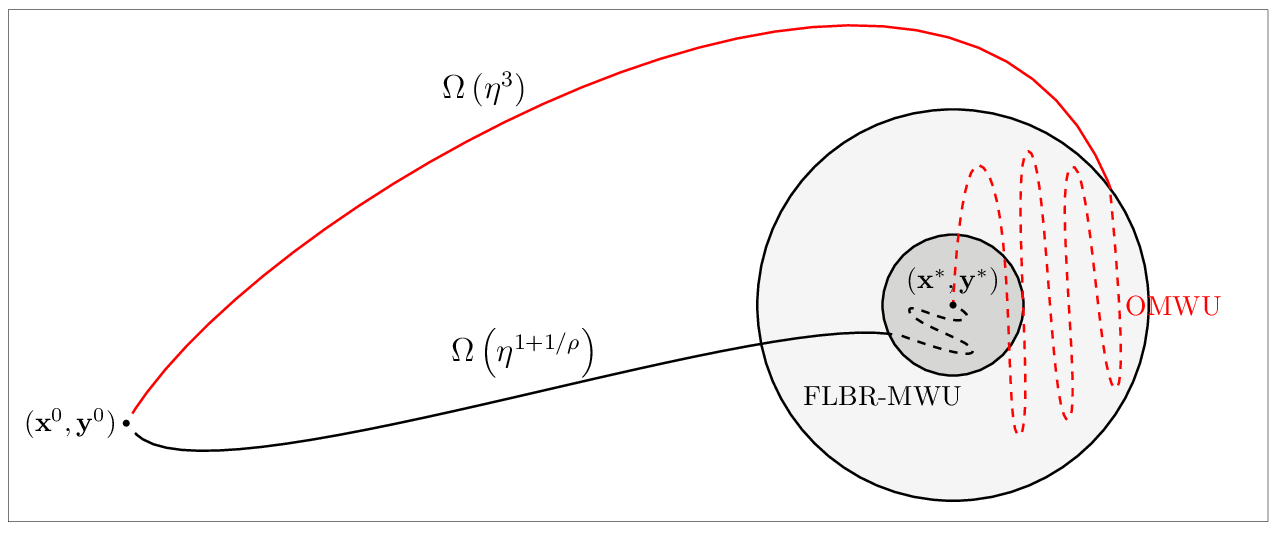}\\
    \includegraphics[width=0.45\textwidth]{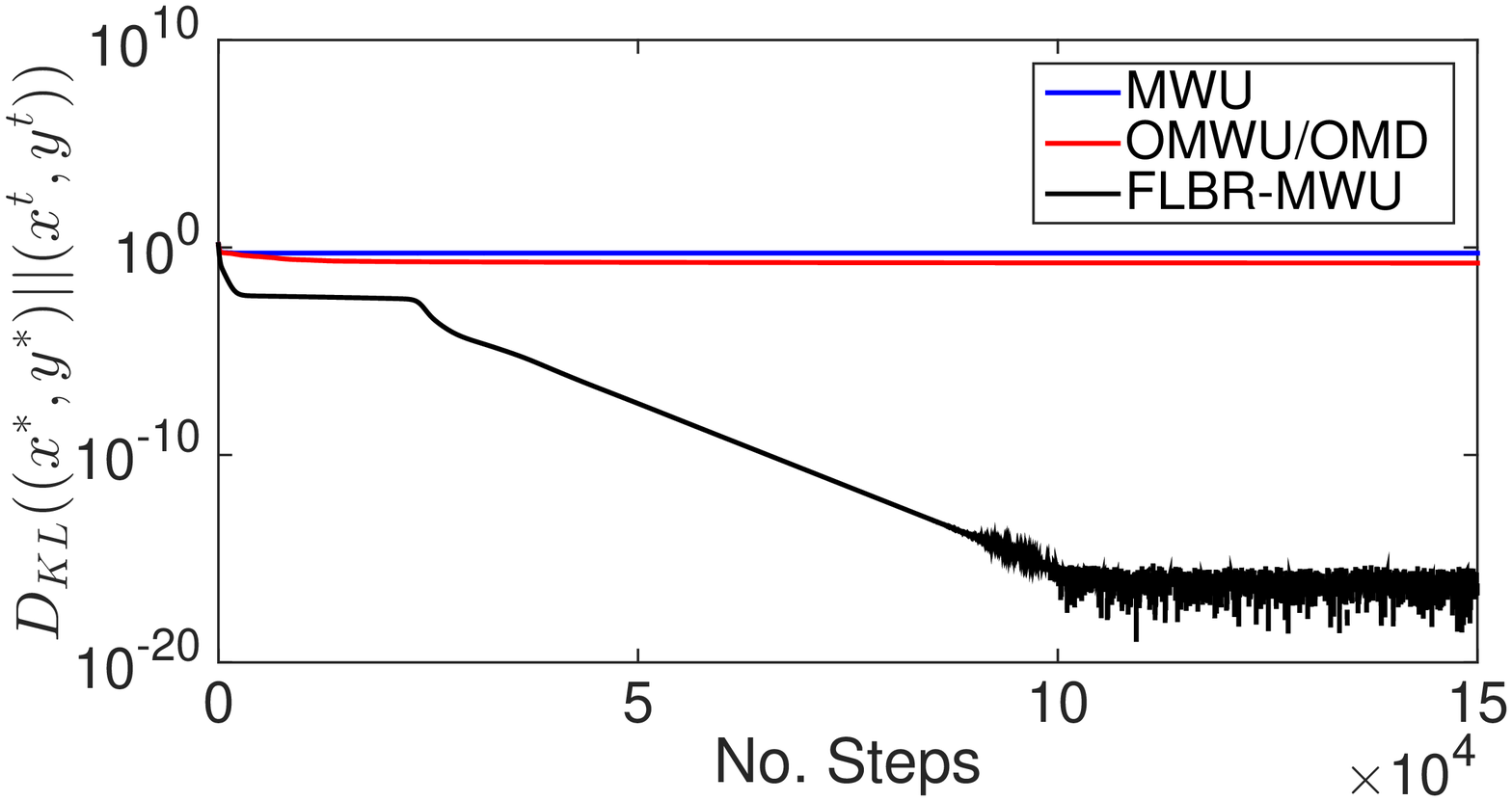}
    \includegraphics[width=0.45\textwidth]{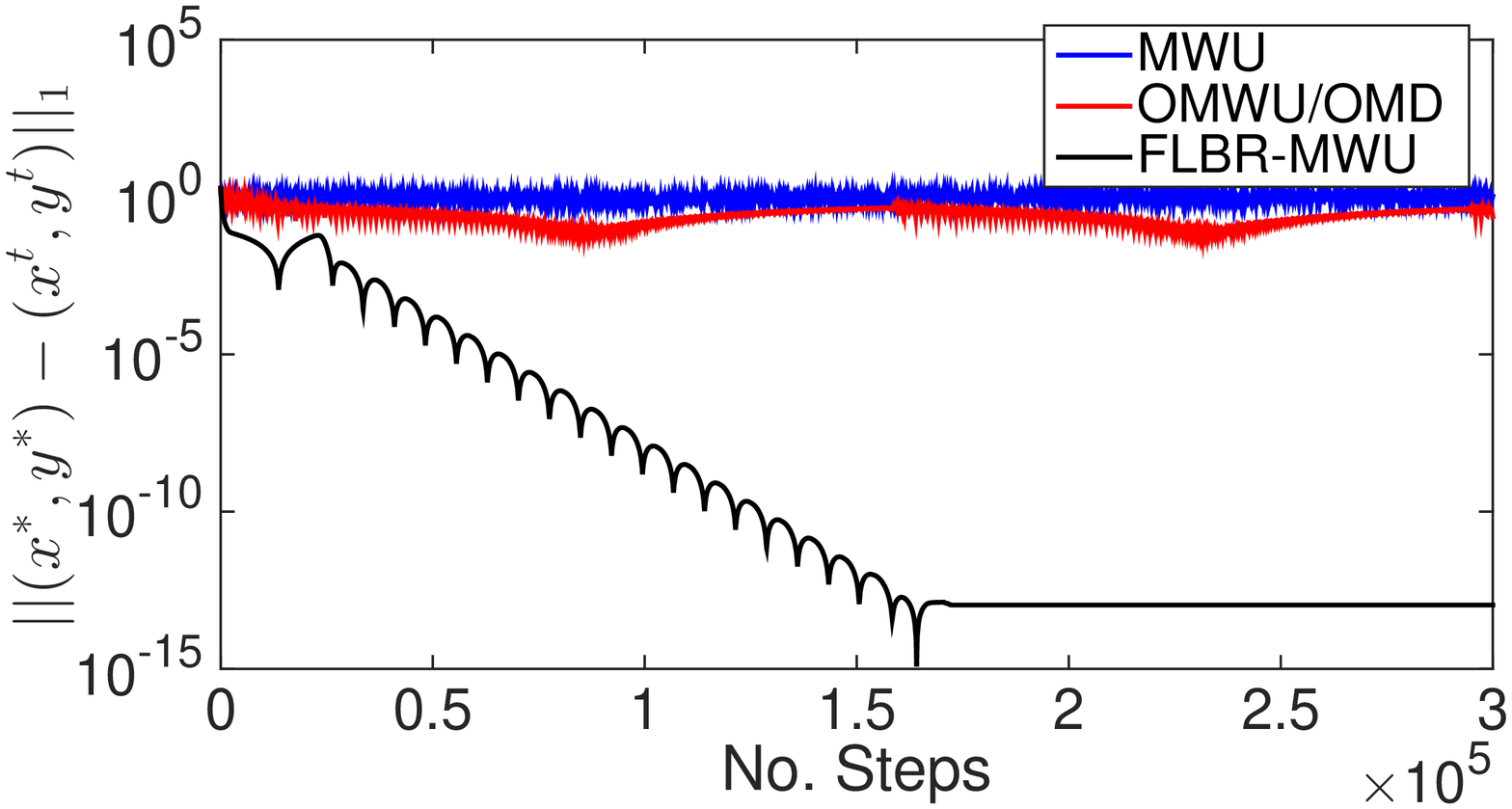}
    \caption{Upper plot: Schematic representation of the convergence path of OMWU (red) and FLBR-MWU (black). Lower plots: A random realization of the learning dynamics for three variants of MWU. 
    The convergence rate for the proposed algorithm (FLBR-MWU) is significantly faster than the existing state-of-the-art algorithms.}
    \label{motiv:ex:fig}
\end{figure*}

The upper plot in Figure \ref{motiv:ex:fig} attempts to demonstrate in a qualitative manner the differences we observed in convergence between the proposed FLBR-MWU and the OMWU dynamics. The two phases of the learning dynamics (decrease of divergence, followed by contraction), are highlighted along with the regions of convergence. A quantitative presentation is shown in the lower plots of Figure \ref{motiv:ex:fig} which depict the convergence behavior of MWU (blue lines), OMWU (red lines) and FLBR-MWU (black lines) for a random realization of a $10\times 10$ payoff matrix with learning rate $\eta=0.1$. We provide two measures of convergence, the Kullback-Leibler Divergence ($D_{KL}$) from the Nash equilibrium (lower left panel), and the respective $l_1$ norm difference (lower right panel), which reveal different aspects of the dynamics.

As anticipated, MWU fails to converge and a smaller learning rate $\eta$ would not fix this issue. OMWU does converge  but in a very slow pace requiring an enormous number of steps, whereas OMD behaves almost in the same manner as OMWU as expected by \cite{Wei2021LinearLC} (OMD is not explicitly shown here nor in Section \ref{sec:exp}; see Section C in the supplementary material for longer, and more detailed simulations). On the other hand, FLBR-MWU converges up to machine precision, as revealed by both $D_{KL}$ and $l_1$ metrics. Indeed, FLBR-MWU is able to escape from the $D_{KL}$ plateau (seen in the lower left panel), where the dynamics are moving towards a direction with slow $D_{KL}$ decline, and ultimately converges to the equilibrium in an oscillatory manner with decreasing amplitude (damped oscillations), as is evident from the $l_1$-norm difference (lower right panel). 
Overall, FLBR-MWU has more than one order of magnitude faster convergence rate relative to OMWU and furthermore tolerates larger values for the learning rate, thus the speed of equilibrium computation is significantly accelerated.


\subsection{Related Work}\label{sec:related work}

To position our paper within the existing literature, the works most related to ours are \cite{Daskalakis2019LastIterateCZ} and \cite{DBLP:conf/iclr/MertikopoulosLZ19}, regarding the OMWU and the OMD methods, respectively. Conceptually, the definition of our method is closer to \cite{DBLP:conf/iclr/MertikopoulosLZ19} since it uses an extra gradient step. Technically, however, our analysis is based on the mathematical arsenal used in \cite{Daskalakis2019LastIterateCZ}. 

We also overview other relevant works on optimization methods for learning problems. 
It is well known that most of the standard no-regret algorithms exhibit convergence only in an average sense. Hence, the solution at iteration $t$, as $t\rightarrow \infty$, may diverge or enter a limit cycle. Such behaviors can be observed, among others, for Gradient Descent/Ascent (GDA) in unconstrained optimization, as well as for MWU methods for constrained problems, see e.g.,  \cite{DBLP:conf/sigecom/BaileyP18}. 
Given the importance of achieving last-iterate convergence for applications on learning, such as training GANs, \cite{daskalakis2018training} and \cite{DBLP:conf/aistats/LiangS19} studied an optimistic variant of GDA, referred to as OGDA, which has also been considered in previous works, 
e.g., \cite{NIPS2013_f0dd4a99}. Their main result states that OGDA exhibits last-iterate convergence for the unconstrained minmax problem with bilinear functions.

Driven by this, \cite{Daskalakis2019LastIterateCZ} proposed to study the constrained version of minmax problems, that corresponds to finding equilibria in zero-sum games. They proposed an optimistic variant of MWU, termed OMWU, and proved that for games with a unique Nash equilibrium, it converges in the last-iterate sense. The sequence of approximations in OMWU uses two previous steps in order to compute the next update, where the extra term corrects the behaviour of the MWU dynamics. 
Moreover, the performance of OMWU provides a strengthening to the supporting experimental evidence in \cite{DBLP:conf/nips/SyrgkanisALS15}. 
Very recently, an analysis on the number of required steps, both for OMWU and OGDA, was provided in \cite{Wei2021LinearLC}, establishing a linear rate of convergence.
Finally, an alternative view on the behavior of OMWU by studying
volume contraction is given in \cite{DBLP:conf/nips/CheungP20},
and further generalizations have been obtained in \cite{DBLP:conf/aistats/LeiNPW21} for convex-concave landscapes. 

In parallel to the study of OMWU, the work of \cite{DBLP:conf/iclr/MertikopoulosLZ19} considered a method, where they agglomerate an intermediate approximation with the former state in order to compute the next state.
This method is known as Optimistic Mirror Descent (OMD) \cite{DBLP:journals/jmlr/ChiangYLMLJZ12,NIPS2013_f0dd4a99}, or Mirror-Prox \cite{Nemirovski2004ProxMethodWR}. The main result of \cite{DBLP:conf/iclr/MertikopoulosLZ19} is that OMD attains last-iterate convergence for a quite general class of problems that encompasses zero-sum games. A similar approach was introduced by \cite{DBLP:conf/iclr/GidelBVVL19} to cope with computational issues. 
All these techniques fall under the umbrella of \emph{extra-gradient} methods, whose origins date back to \cite{Korpelevich1976TheEM} (for more details see \cite{Facchinei2003FiniteDimensionalVI,Bubeck15}). 

Several other streams of works have also focused on convergence properties of extra-gradient methods. 
In \cite{DBLP:conf/aistats/LiangS19}, a linear convergence rate was proved for OMD under the assumption that the game matrix is square and full rank. Under the same assumption, the work of \cite{DBLP:conf/aistats/MokhtariOP20} considered more general forms of saddle-point problems. 
Recently, in a different direction, \cite{Cen2021FastPE} developed extra-gradient methods for Quantal Response equilibria, which can be used also for finding an approximate Nash equilibrium.
Furthermore, \cite{Hsieh2019OnTC} investigated asymptotic last-iterate convergence for variants of extra-gradient algorithms. 

To our knowledge, the idea of using different rates in the intermediate and the update steps of extra gradient methods, has also been used in \cite{Azizian}. There are however substantial differences with our work. Most importantly, \cite{Azizian} involves the unconstrained bilinear case. Even further, their result holds under certain spectral assumptions, and shows only local convergence (starting from a point near the fixed point), whereas we do not need such a condition.  


\section{Basic Definitions}\label{sec:Basic definitions}

\subsection{Zero-sum Games and Approximate Equilibria}
\label{subsec:Bimatrix games}

We consider finite 2-player zero-sum games, defined by a matrix\footnote{We can always scale appropriately so that the entries are in $(0, 1]$, without affecting the equilibrium strategies.} $R\in(0,1]^{n \times n}$, where without loss of generality, we assume both players have $n$ \emph{pure} strategies. We refer to the two players as the {\it row player} and the {\it column player} respectively. 
If the row player plays the $i$-th row and the column player plays the $j$-th column, then the payoff of the row player is $R_{ij}$, and the payoff of the column player is $-R_{ij}$. 
We also allow \emph{mixed strategies} as probability distributions (column vectors) on the pure strategies. E.g., a mixed strategy for the row player will be denoted as $x = (x_1,\dots, x_n)$, where $x_i$ is the probability of playing the $i$-th row. For convenience, we will denote the $i$-th pure strategy of a player by the unit vector $e_i$, which has probability one in its $i$-th coordinate and 0 elsewhere. 

A pair $(x,y)$, where $x, y$ are mixed strategies for the row and the column player respectively, is called a strategy profile. Given such a profile, the expected payoff of the row player is $x^TRy$, whereas for the column player, it is $-x^TRy$. This is obviously a bilinear function, since it is equivalent to $\sum_{i, j} R_{ij}x_i y_j$. 

The fundamental solution concept in game theory is that of Nash equilibrium, stating that no player has an incentive to deviate to another strategy.  

\begin{definition} 
A strategy profile $(x^*,y^*)$ is a Nash equilibrium in the zero-sum game defined by matrix $R$, if and only if, for any $i,j\in [n]$,
\begin{equation*}
(x^*)^T R y^* \geq e_i^T R y^* \text{ and } (x^*)^T R e_{j}\geq (x^*)^T R y^*,
\end{equation*}
\end{definition}

The payoff of the row player at an equilibrium, $v = (x^{*})^TRy^*$, is referred to as the \emph{value} of the game.
It is well known that the value of a game and its equilibrium strategies are given by the solution of the following max-min (saddle-point) problem, over the $n$-dimensional simplex $\Delta_n$: $v = \max_{x \in \Delta_n} \min_{y \in \Delta_n} x^TRy = \min_{y \in \Delta_n}\max_{x \in \Delta_n} x^TRy$.
 
A useful concept in the analysis of games is the \emph{support} of a mixed strategy $x$, which is the set of pure strategies that have a positive probability under $x$, i.e., $supp(x) = \{i: x_i>0\}$. 
It is easy to see that at an equilibrium $(x^*,y^*)$, any pure strategy $e_i$, with $i\in supp(x^*)$, is a best response against $y^*$ (resp. for any $j\in supp(y^*)$, $e_j$ is a best response to $x^*$).

In our work, we will also need to argue about approximate equilibria to establish convergence. We start with defining approximate best responses. Given a profile $(x, y)$, we say that a strategy $x'$ is an \emph{$\varepsilon$-best-response} strategy to $y$ with $\varepsilon \in [0,1]$, if it yields a payoff that is at most $\varepsilon$ less than the best-response payoff. We can define now an approximate equilibrium, as a profile $(x, y)$ where $x$ and $y$ are both approximate best responses to each other. This is precisely the standard notion of additive, approximate equilibria \cite{NRTV07}. 

\begin{definition}
A strategy profile $(x^*,y^*)$ is an $\varepsilon$-Nash equilibrium in the zero-sum game defined by matrix $R$, if and only if, for any $i,j$
\begin{equation*}
    (x^*)^T R y^* \geq e_i^T R y^* - \varepsilon \enspace \text{and} \enspace (x^*)^T R y^* \leq (x^*)^T R e_{j} + \varepsilon
\end{equation*}
\end{definition}


\subsection{Relevant MWU Variants}

One of the standard versions of multiplicative weights update methods results from the FTRL dynamics (Follow-The-Regularized-Leader), when the regularizer is the negative entropy function, (see e.g., \cite{Hoi18}).
MWU rewards better the pure strategies that perform well against the previous iteration. 
In particular, if $(x^{t-1}, y^{t-1})$ is the profile at the end of iteration $t-1$, and $\eta$ is the learning rate parameter, then for $i\in [n]$, $x^t_i$ is set to be analogous to 
$x^{t-1}_i \cdot e^{\eta e_i^T R y^{t-1}}$ (with appropriate normalization). 

In the remaining paper, we often make comparisons or references to the optimistic variant proposed by \cite{Daskalakis2019LastIterateCZ}, referred to as OMWU. The idea of "optimism" here is 
to take into account two previous iterations in order to compute the next update, where the extra term can be seen as a negative momentum, correcting the behaviour of MWU dynamics.
The dynamics of OMWU are described below for all $i,j \in [n]$.
\begin{equation}\label{eq:OMWU}
\begin{split}
 & x^t_i = x^{t-1}_i \cdot \frac{e^{2\eta e_i^T R y^{t -1} - \eta e_i^T R y^{t -2}}}{\sum\nolimits_{j=1}^{n} x^{t-1}_j e^{2 \eta e_j^T R y^{t-1}-\eta e_j^T R y^{t -2}}}, \enspace y^t_j = y^{t-1}_j \cdot \frac{e^{-2\eta e_j^T R^T x^{t-1} + \eta e_j^T R^T x^{t-2}}}{\sum\nolimits_{i=1}^{n} y^{t-1}_i e^{-2 \eta e_i^T R^T x^{t-1} + \eta e_i^T R^T x^{t-2}}}.
\end{split}
\end{equation}


\section{Forward Looking Best-Response Multiplicative Weights Update Method (FLBR-MWU)}
\label{sec:theory}

\subsection{Definition of the Dynamics}
\label{subsec:def}

We now present the method studied in this work, which we refer to as Forward Looking Best-Response Multiplicative Weights Update method (FLBR-MWU). We provide first a short description of the main idea behind the dynamics. This is an extra gradient method and each iteration has an intermediate and a final step. Suppose that starting from some initial profile, we reach the profile $(x^{t-1}, y^{t-1})$ by the end of iteration $t-1$. In the intermediate step of iteration $t$, we compute a strategy $\hat{x}^t$ for the row player (resp. $\hat{y}^t$ for the column player), which is an approximate best-response strategy to $y^{t-1}$ (resp. to $x^{t-1}$). This serves as a look ahead step of what would be the currently optimal choices. In the final step of iteration $t$, we compute the new mixed strategy $x^t$ for the row player, by 
performing multiplicative weights updates, but after assuming that the opponent was playing $\hat{y}^t$. 

Formally, the first step of the dynamics, denoted as the intermediate best response (IBR) step, is defined below, at iteration $t$, and for all $i, j\in [n]$, given a non-negative parameter $\xi\in\mathbb{R}^+$ ($\xi$ will be chosen sufficiently large, as will become clear from Lemma \ref{lem:xi}).
\begin{equation}
\label{step1}
\begin{split}
 & \hat{x}^t_i = x^{t-1}_i\cdot\frac{e^{\xi e_i^TRy^{t-1}}}{\sum\nolimits_{j=1}^{n}x^{t-1}_j e^{\xi e_j^TRy^{t-1}}}, \enspace \hat{y}^t_j = y^{t-1}_j\cdot\frac{e^{-\xi e_j^TR^Tx^{t-1}}}{\sum\nolimits_{i=1}^{n}y^{t-1}_i e^{-\xi e_i^TR^Tx^{t-1}}}.
\end{split}
\end{equation}

The second step, which updates the profile $(x^{t-1}, y^{t-1})$ to $(x^{t}, y^{t})$ is below, given the learning rate parameter $\eta\in (0, 1)$. We assume that we use the same fixed constants $\eta$ and $\xi$ in all iterations\footnote{It is an interesting topic for future work, to examine adaptive schemes for $\xi$ and $\eta$ throughout the iterations.}.
\begin{equation}
\label{step2}
\begin{split}
 & x^t_i = x^{t-1}_i\cdot\frac{e^{\eta e_i^TR\hat{y}^{t}}}{\sum\nolimits_{j=1}^{n}x^{t-1}_j e^{\eta e_j^TR\hat{y}^{t}}}, \enspace y^t_j = y^{t-1}_j\cdot\frac{e^{-\eta e_j^TR^T\hat{x}^{t}}}{\sum\nolimits_{i=1}^{n}y^{t-1}_i e^{-\eta e_i^TR^T\hat{x}^{t}}}.
\end{split}
\end{equation}

\begin{remark}
By setting $\xi = \eta$ in Equation \eqref{step1} above, the proposed method becomes the same as OMD with entropic regularization \cite{DBLP:conf/iclr/MertikopoulosLZ19}, which can also be viewed as OMWU \cite{Wei2021LinearLC}. In our method however, $\eta$ and $\xi$ differ substantially across both our theoretical and experimental results.
\end{remark}


\subsection{Main Results}

We consider games with a unique Nash equilibrium, as in \cite{Daskalakis2019LastIterateCZ}, since it has been argued that the set of zero-sum games with non-unique equilibrium has \emph{Lebesgue measure} equal to zero \cite{Damme91}. For convenience, we also assume that the initial strategy profile consists of the uniform distribution for each player. However, our results hold for any fully-mixed initial profile, with a non-zero probability to all pure strategies. 

The main result of our work is the following theorem. 

\begin{theorem}
\label{thm:main}
Consider a zero-sum game with a unique Nash equilibrium $(x^*, y^*)$. Starting with the uniform distribution for each player, the FLBR-MWU dynamics attain last-iterate convergence to the Nash equilibrium, i.e., $\lim_{t\rightarrow \infty} (x^t, y^t) = (x^*, y^*)$, when $\eta$ is chosen sufficiently small, and for big enough $\xi$, so that $\eta \xi <1$.
\end{theorem}

The goal of the remaining section is to establish the proof of Theorem \ref{thm:main}. Towards this, we start with the choice of $\xi$. The next lemma provides the important observation, that as $\xi \rightarrow \infty$, the strategy $\hat{x}^t$, computed in the first step of iteration $t$, becomes a best response against $y^{t-1}$ (analogously for $\hat{y}^t$). 

\begin{lemma}\label{lem:xi}
Given any $t > 0$, let $\hat{x}^t$, $\hat{y}^t$,  be the strategies produced by the first step of iteration $t$. As 
$\xi \to +\infty$, then $\hat{x}^t$ converges to a best-response strategy against $y^{t-1}$ (similarly for $\hat{y}^t$ against $x^{t-1}$).
\end{lemma}

The proof of the lemma can be found in the supplementary material. In the sequel, we assume that $\xi$ has been chosen sufficiently large, so that $\hat{x}^t$ is an $\epsilon$-best response with $\epsilon \rightarrow 0$. 
For appropriate choices of $\xi$ in practice, we refer to the discussion in Section \ref{sec:exp}. 

The proof of Theorem \ref{thm:main} is split into 3 parts. The first part establishes that after a certain number of iterations, the dynamics reach a profile $(x^t, y^t)$, that is an $O(\eta^{1/\rho})$-Nash equilibrium with $\rho > 1$. The second part shows that the profile $(x^t, y^t)$ lies within a neighborhood of the actual equilibrium $(x^*, y^*)$. Finally, the last part shows that the update rule of FLBR-MWU is a contracting map, i.e., once we are within a neighborhood of $(x^*, y^*)$, the dynamics converge to their fixed point, which directly implies last-iterate convergence. These three parts are established in Theorems \ref{thm:KL_divergence}, \ref{thm:neighborhood} and \ref{thm:contraction} respectively. The structure of the proof is similar to the convergence proof of OMWU in \cite{Daskalakis2019LastIterateCZ}. There are however differences in various parts of the analysis. Most importantly, in the first part, we are able to establish a better convergence rate to an approximate equilibrium, whereas OMWU achieves an $\Omega(\eta^3)$ decrease rate. Furthermore, in the third part, the analysis of our Jacobian matrix (proof of Theorem \ref{thm:contraction}) is also different since we are analyzing sufficiently different dynamics.  

To proceed with the first part of the proof, we will use the \emph{Kullback-Leibler (KL) divergence} as a measure of progress. The KL divergence quantifies the similarity between two distributions, and here we will consider the divergence between a profile $(x^t, y^t)$ and the equilibrium $(x^*,y^*)$, which equals:
\begin{equation}\label{eq:KL Divergence}
\begin{split}
 & D_{KL}((x^*,y^*)||(x^t,y^t)) = \sum\limits_{i=1}^{n}x^*_i \ln(x^*_i/x^t_i) + \sum\limits_{j=1}^{n}y^*_j \ln(y^*_j/y^t_j).
\end{split}
\end{equation}
Note that by the initialization and the definition of the dynamics, $x_i^t > 0$, $y_j^t > 0$ for any given $t$, and any $i$, $j$, so that the logarithmic terms above are well-defined.

\begin{theorem}\label{thm:KL_divergence}
Consider a zero-sum game with a unique Nash equilibrium $(x^*,y^*)$. Assume that we run the FLBR-MWU dynamics with the uniform distribution as the initial strategy for both players, and using a sufficiently small $\eta$ and a big enough $\xi$. Then, for any $\rho>1$, the KL divergence $D_{KL}((x^*,y^*)||(x^t,y^t))$ decreases at every iteration with a rate of at least $\Omega{(\eta^{1+1/\rho})}$, until we reach an $O(\eta^{1/\rho})$-Nash equilibrium of the game.
\end{theorem}

\begin{proof}
Let $(x^*,y^*)$ be the Nash equilibrium of the game, and let $v$ be the value of the game, $v = (x^{*})^TRy^*$. We take the difference of the KL divergences between two consecutive iterations:
\begin{equation*}
\begin{split}
 & D_{KL}((x^*,y^*)||(x^t,y^t)) - D_{KL}((x^*,y^*)||(x^{t-1},y^{t-1})) = \\ 
 & \hspace{5cm} -\Big(\sum\nolimits_{i=1}^{n} x^*_i \ln(x^{t}_i/x^{t-1}_i) + \sum\nolimits_{j=1}^{n} y^*_j \ln(y^{t}_j/y^{t-1}_j)\Big ).
\end{split}
\end{equation*}
We show that this difference is negative and we quantify the decrease in the KL divergence, till we reach an $O(\eta^{1/\rho})$-Nash equilibrium. Analytically, we have that
\begin{equation*}
\begin{split}
  D_{KL} & ((x^*,y^*)||(x^t,y^t)) - D_{KL}((x^*,y^*)||(x^{t-1},y^{t-1})) \\
 &  = -\sum\limits_{i=1}^{n}x^*_i \ln e^{\eta e_i^T R \hat{y}^t} + \ln\Big(\sum\limits_{i=1}^{n} x^{t-1}_i e^{\eta e_i^T R \hat{y}^t}\Big) - \sum\limits_{j=1}^{n}y^*_j \ln e^{-\eta e_j^T R^T \hat{x}^t} + \ln\Big(\sum\limits_{j=1}^{n} y^{t-1}_j e^{-\eta e_j^T R^T \hat{x}^t}\Big) \\
 & = -\eta x^{*T}R\hat{y}^t + \eta (y^{*})^T R^T \hat{x}^t +\eta (x^{t-1})^T R y^{t-1} -\eta (y^{t-1})^T R^T x^{t-1} + \ln\Big(\sum\limits_{i=1}^{n} x^{t-1}_i e^{\eta e_i^T R \hat{y}^t - \eta (x^{t-1})^T R y^{t-1}}\Big)  \\
 & \hspace{1cm} + \ln\Big(\sum\limits_{j=1}^{n} y^{t-1}_j e^{-\eta e_j^T R^T \hat{x}^t+\eta (y^{t-1})^T R^T x^{t-1}}\Big).
\end{split}
\end{equation*}

Notice that in the last expression above, the third term ($\eta (x^{t-1})^TRy^{t-1}$) cancels out with the fourth term. Also, since $(x^*,y^*)$ is an equilibrium, it holds that $x^{*T}R\hat{y}^t\geq v$ and $(y^{*})^TR^T\hat{x}^t \leq v$.
Therefore, the first and second terms also cancel out and yield an upper bound with the two logarithmic terms. 

We now apply the Taylor expansion of $e^x$. For convenience, let
$p_i(\eta) = \eta(e_i^TR\hat{y}^t - (x^{t-1})^TRy^{t-1})$, and let 
$q_j(\eta) = \eta(-e_j^TR^T\hat{x}^t + (x^{t-1})^TRy^{t-1})$. Using these abbreviations, the difference of the KL divergences is upper bounded by
\begin{equation*}
\begin{split}
 & \ln \Big(1+\eta ((x^{t-1})^T R \hat{y}^t - (x^{t-1})^T R y^{t-1}) + \sum\limits_{i=1}^{n} x^{t-1}_i \sum\limits_{k=2}^{\infty} \frac{(p_i(\eta))^k}{k!} \Big) \\
 & \hspace{3cm} + \ln \Big(1 + \eta (-(y^{t-1})^T R^T \hat{x}^t   + (y^{t-1})^T R^T x^{t-1}) + \sum\limits_{j=1}^{n} y^{t-1}_j \sum\limits_{k=2}^{\infty} \frac{(q_j(\eta))^k}{k!} \Big).
\end{split}
\end{equation*}

It is easy to see that $|p_i(\eta)| \leq \eta$ and $|q_j(\eta)| \leq \eta$. This means that for any $k\geq 2$ (i.e., for both odd and even values of $k$), $(p_i(\eta))^k \leq \eta^k$ and $(q_j(\eta))^k \leq \eta^k$. By using the geometric series, we have that $\sum\limits_{k=2}^{\infty} \frac{(p_i(\eta))^k}{k!} \leq \eta^2/(1-\eta)$, and similarly for the series concerning $q_j(\eta)$. If we also use the inequality $\ln(x)\leq x - 1$, we obtain the following sequence of steps.
\begin{equation}\label{ineq:DKL_dist}
\begin{split}
    & D_{KL} ((x^*,y^*)||(x^t,y^t)) - D_{KL}((x^*,y^*)||(x^{t-1},y^{t-1})) \\
    & \hspace{2cm} \leq \ln \Big(1+\eta ((x^{t-1})^T R \hat{y}^t - (x^{t-1})^T R y^{t-1}) + \frac{\eta^2}{(1-\eta)} \Big) \\
    & \hspace{4cm} + \ln \Big(1 + \eta (-(y^{t-1})^TR^T\hat{x}^t+(y^{t-1})^T R^T x^{t-1})  + \frac{\eta^2}{(1-\eta)} \Big) \\
    & \hspace{2cm} \leq \eta ((x^{t-1})^T R\hat{y}^t - (x^{t-1})^T Ry^{t-1} + (x^{t-1})^T R y^{t-1} - \hat{x}^t R y^{t-1}) + 2 \frac{\eta^2}{(1-\eta)} \\ 
    & \hspace{2cm} = -\eta(\varepsilon_1 +\varepsilon_2) +  4\eta^2  \leq -\eta (\max\{\varepsilon_1,\varepsilon_2\}) +4 \eta^2  ,
\end{split}
\end{equation}

where $\varepsilon_1= \hat{x}^t R y^{t-1} - (x^{t-1})^T R y^{t-1}$, $\varepsilon_2= (x^{t-1})^T R \hat{y}^t - (x^{t-1})^T R y^{t-1}$, and the last inequality holds because $\eta \leq 1/2$. Let us look now more carefully at $\varepsilon_1$ (an analogous argument holds for $\varepsilon_2$). The term $\varepsilon_1$ expresses the additional benefit for the row player, if at the profile $(x^{t-1}, y^{t-1})$, she deviates to $\hat{x}^t$. By Lemma \ref{lem:xi}, we know that as $\xi\rightarrow \infty$, then $\hat{x}^t$ tends to her best response against $y^{t-1}$. Hence when we select $\xi$  sufficiently large, $\varepsilon_1$ tends to the best possible deviation gain of the row player at the profile $(x^{t-1}, y^{t-1})$ (resp. for $\varepsilon_2$ and the column player).

To finish the proof, suppose that the profile $(x^{t-1},y^{t-1})$ is not an $O(\eta^{1/\rho})$-Nash equilibrium. Then there exists a deviation that provides additional gain of $\Omega(\eta^{1/\rho})$ to one of the players. This implies that $\max\{\varepsilon_1, \varepsilon_2\} = \Omega(\eta^{1/\rho})$. 
Hence, by (\ref{ineq:DKL_dist}), and since $\eta<1$, we can see that as long as we have not reached an $O(\eta^{1/\rho})$-Nash equilibrium, the KL divergence will keep decreasing by at least  $\eta \Omega(\eta^{1/\rho}) - 4 \eta^2 = \Omega(\eta^{1+1/\rho})$.
As the KL divergence cannot decrease forever, eventually, our dynamics will reach an $O(\eta^{1/\rho})$-Nash equilibrium. 
\end{proof}

Consider now the first iteration $t$ of the dynamics, where $(x^t, y^t)$ forms an $O(\eta^{1/\rho})$-Nash equilibrium for some fixed $\rho > 1$. The next step is to show that if we make $\eta$ small enough, this profile falls within a neighborhood of the equilibrium $(x^*, y^*)$.

\begin{theorem}
\label{thm:neighborhood}
Let $(x^*,y^*)$ be the unique Nash equilibrium of the zero-sum game, and let $(x^t, y^t)$ be the first profile reached by the dynamics, that is an $O(\eta^{1/\rho})$-Nash equilibrium for some $\rho>1$. Then 
\begin{equation*}
    \lim\limits_{\eta\rightarrow 0} ||(x^*, y^*) - (x^t, y^t)||_1 = 0,
\end{equation*}
\end{theorem}

The proof of Theorem~\ref{thm:neighborhood} can be found at the supplementary material, Appendix~\ref{supsec:Omitting proofs}.

The next and final step of our proof is to show that our dynamics induce a contracting map. An update rule with a fixed point $x$ is called a contraction, if there exists a region $U$ around $x$, such that for any starting point in $U$, the rule converges to its fixed point as $t\rightarrow \infty$. 
In our case, the Nash equilibrium $(x^*, y^*)$ of the game is a fixed point of the FLBR-MWU dynamics and Theorem \ref{thm:neighborhood} guarantees that we can reach a neighborhood around $(x^*, y^*)$.
To proceed, we state a sufficient condition for a dynamical system to converge to its fixed point.

\begin{theorem}[see \cite{Galor}]\label{thm:App_jac_contr_thm}
Let ${x}^*$ be a fixed point for the dynamical system ${x}^{(t+1)} = g({x}^{(t)})$. If all eigenvalues of the Jacobian matrix of $g$ at ${x}^*$ have absolute value less than one, then there exists a neighborhood $U$ of $x^*$ such that for all $x \in U$, $g$ converges to $x^*$, starting from $x$.
\end{theorem}
\vspace{-1mm}
Using Theorem \ref{thm:App_jac_contr_thm}, we show the following theorem, whose proof can be found at the supplementary material, Appendix~\ref{supsec:Omitting proofs}.

\begin{theorem}
\label{thm:contraction}
The update rule of FLBR-MWU is a contraction, as long as $\eta\xi < 1$, i.e. $\lim_{t\rightarrow \infty} (x^{\scriptstyle t},y^t) = (x^*,y^*)$.
\end{theorem}


\vspace{-1mm}
\section{Numerical Experiments}
\vspace{-2mm}
\label{sec:exp}
We note that additional supporting figures and elaboration on our experiments can be found in the supplementary material (Appendix~\ref{supsec:Numerical Experiments}).

{\bf Nash equilibrium estimation.} In order to make comparisons, we need first to compute the equilibria of the generated instances. 
Instead of using a linear programming solver, the equilibrium computation is performed using the proposed FLBR-MWU algorithm with $\eta=0.05$. FLBR-MWU is an iterative approach thus a convergence criterion to ensure that the Nash equilibrium has been reached is required. We propose as a convergence criterion the $D_{KL}$ between the update step and the IBR step of our dynamics: $D_{KL}((x^t,y^t)||(\hat{x}^t,\hat{y}^t))$. This metric is sufficient because the best response strategy at Nash equilibrium is exactly the equilibrium strategy, thus $\lim_{t\to\infty} D_{KL}((x^t,y^t)||(\hat{x}^t,\hat{y}^t)) = 0$ (for small enough $\eta$). We return the solution when the convergence criterion becomes $10^{-15}$, which is approximately the machine's arithmetic precision, or when the maximum number of steps, denoted by $t_{\max}$, --typically millions of steps-- has been reached. In the infrequent latter case (it happened in less than 0.1\% of the instances), we discard the returned solution.

{\bf Effect of the intermediate rate ($\xi$).} In our learning dynamics, the best response strategy is approximated by the softmax function (a.k.a. the normalized exponential function or the Gibbs measure in statistical physics). Sending $\xi$ to infinity, one out of the potentially many best response strategies is obtained as intermediate dynamics by Equation \eqref{step1}. However, $\xi$ should be finite from a practical point of view. Since it appears at the exponentials' argument, very high values of $\xi$ may result in arithmetic imprecision. Therefore, we conducted a numerical study to assess the effect of $\xi$ on the convergence of the algorithm. Table \ref{tab:ksi} presents various statistics about the number of steps required for several values of $\xi$ and for two values of the size of the payoff matrix $R$, with $\eta=0.1$. We average over $10^3$ repetitions using random payoff matrices, whose elements are iid sampled from $\mathcal U([0,1])$.
Evidently, as $\xi$ increases, the FLBR-MWU dynamics require fewer steps in order to reach a specific threshold of accuracy (set to $10^{-10}$ for the $D_{KL}$ between the Nash equilibrium and the FLBR-MWU dynamics). However, the solution occasionally produces `NaN' for values of $\xi$ above 200, due to overflow in the exponentials\footnote{Overflow can be easily fixed by subtracting the maximum value but with an increased underflow risk.}. Overall, values between 50 and 100 are a sufficient compromise between the best response approximation and machine precision trade-off. We also note that even when we select values that violate the condition $\eta\xi < 1$ (which we needed for our theoretical results), we still attain convergence in most cases. In the remaining experiments of this section, we set $\xi=100$, even though larger values can be tolerated especially when both $n\gg1$ and $x_i^*,y_j^*\ll1$ hold.

\begin{table}[th]
\caption{Statistics on the number of steps till convergence for various values of $\xi$ and $n$. The maximum number of steps was set to $t_{\max} = 10^6$.}
\label{tab:ksi}
\centering
\begin{tabular}{cccccc}
Matrix size & Statistic & $\xi=20$ & $\xi=50$ & $\xi=100$ & $\xi=200$ \\ \hline 
\multirow{3}{*}{$n=10$} & Mean & 85.9K & 52.7K & 41.6K & 57.6K \\
& Median & 32.0K & 22.3K & 19.8K & 18.0K \\
& $t_{\max}$ was hit & 0.9\% & 0.0\% & 0.0\% & 2.2\% \\ \hline 
\multirow{3}{*}{$n=20$} & Mean & 352.1K & 233.3K & 173.6K & 141.4K \\
& Median & 225.0K & 123.4K & 82.2K & 65.6K \\
& $t_{\max}$ was hit & 13.5\% & 5.9\% & 3.5\% & 2.5\% \\
\end{tabular}
\end{table}

{\bf Effect of the learning rate ($\eta$).} 
The first row of panels in Figure \ref{effect:eta:fig} shows the $D_{KL}$ between the Nash equilibrium and the FLBR-MWU dynamics for the same payoff matrix instance as in Figure \ref{motiv:ex:fig} and for various values of the learning rate, $\eta$. The difference between the left and right panels is that for the right column of panels, the x-axis has been rescaled by multiplying each run with the respective learning rate. A linear scaling is numerically observed showing that the number of steps is effectively of order $O(\eta^{-1})$ for a fixed accuracy level. This inversely-proportional behavior is observed not only during the convergence to the approximate Nash equilibrium, but also during the contraction period. As a rule of thumb, we propose to increase the rate $\eta$, because it accelerates the convergence, but with caution since a very large $\eta$ might result in an oscillatory solution, thus failing to converge (blue line in second row of panels).

\begin{figure}[h]
    \centering
    \includegraphics[width=0.47\textwidth,height=0.17\textheight]{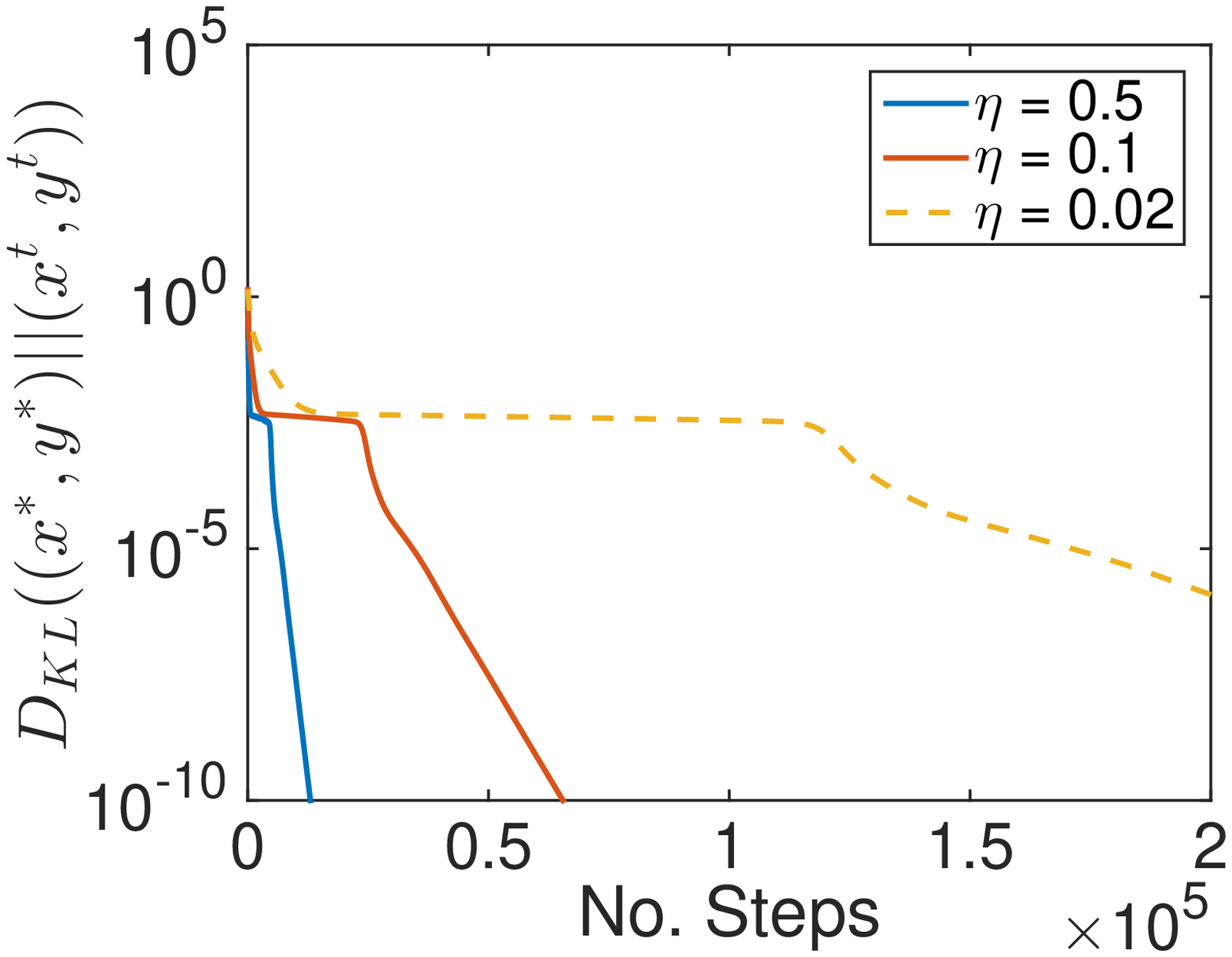}
    \includegraphics[width=0.47\textwidth,height=0.17\textheight]{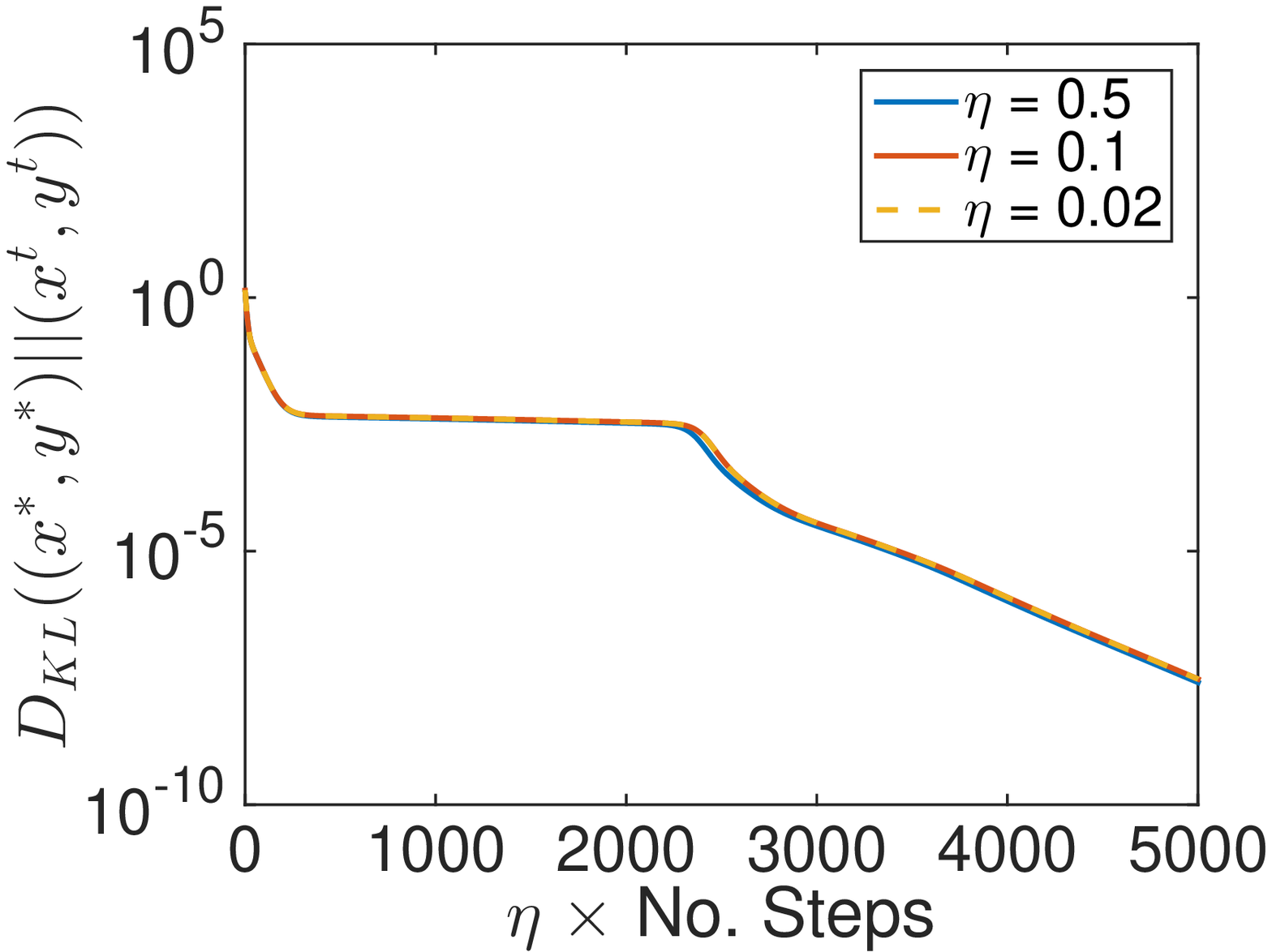}
    \includegraphics[width=0.47\textwidth,height=0.17\textheight]{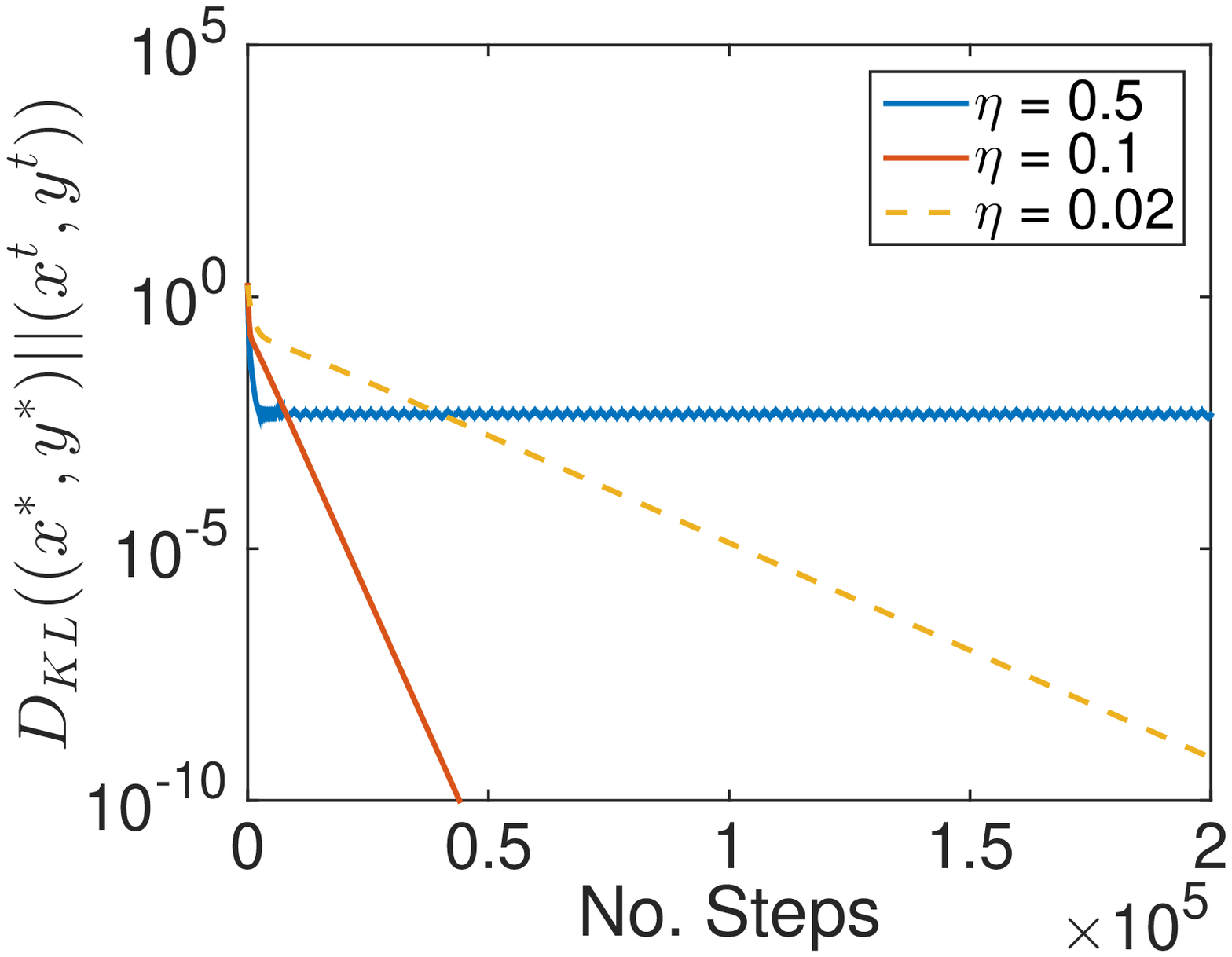}
    \includegraphics[width=0.47\textwidth,height=0.17\textheight]{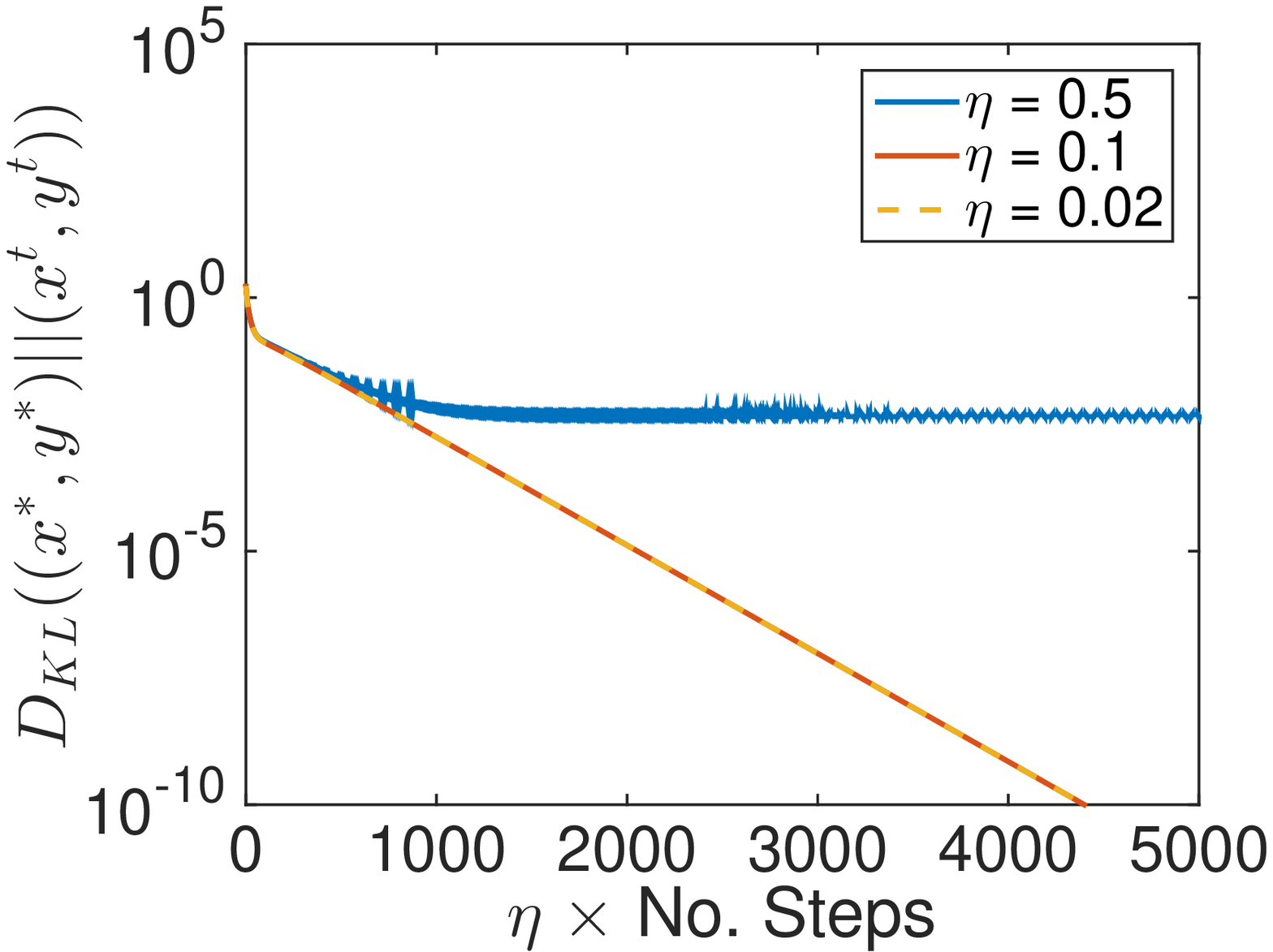}
    \caption{The $D_{KL}$ between the Nash equilibrium and the FLBR-MWU dynamics for two instances and no rescaling of x-axis (left panels) and with rescaling (right panels). The relationship between number of steps and learning rate are inversely proportional.}
    \label{effect:eta:fig}
\end{figure}

{\bf Effect of the payoff matrix size ($n$).} The rate of convergence is sensitive to the size of the payoff matrix and the number of steps is expected to substantially increase on average as the size of the game increases. We performed a numerical comparison between FLBR-MWU and OMWU to evaluate the number of steps required to achieve a predefined level of accuracy.
Table \ref{tab:size} presents statistics on the number of steps for each learning algorithm computed on 100 repetitions using element-wise uniformly-sampled and iid  random payoff matrices. The learning rate was set to $\eta=0.1$. Given that FLBR-MWU requires almost twice as many calculations per iteration, relative to  OMWU, it is fair to multiply the number of steps of FLBR-MWU with two and then compare it with the number of steps of OMWU. We observe that FLBR-MWU is approximately 15 times faster on average when $n=5$. As the size of the payoff matrix increases, the performance gap in convergence rate as measured by the number of steps also increases. Indeed, even for $n=10$, OMWU requires more than $4.2M$ steps in half of the runs, while the respective number for FLBR-MWU is $16.3K$, implying that FLBR-MWU is 100 times faster than OMWU in the median sense. Larger game sizes make OMWU essentially impractical while FLBR-MWU is still able to converge in less than $5M$ steps.

\begin{table}[ht]
\caption{Statistics on the number of steps till convergence for various sizes of the game. The maximum number of steps was set to $t_{\max} = 5\times 10^6$.}
\label{tab:size}
\centering
\begin{tabular}{ccccc}
Learning alg. & Statistic & $n=5$ & $n=10$ & $n=50$ \\ \hline
\multirow{3}{*}{FLBR-MWU} & Mean & 33.7K & 103.3K & 984.9K \\
& Median & 9.8K & 16.3K & 409.3K  \\
& $t_{\max}$ was hit & 0.0\% & 0.0\% & 1.0\%  \\ \hline
\multirow{3}{*}{OMWU} & Mean &  1088.8K & 3323.2K & 5000.0K \\
& Median & 353.8K & 4208.1K & 5000.0K \\
& $t_{\max}$ was hit & 9.0\% & 46.0\% & 100.0\%  \\
\end{tabular}
\end{table}

\newpage


\bibliographystyle{plain}
\bibliography{FLBRMWU}

\newpage


\appendix

\section{Appendix}\label{sec:Appendix}

\subsection{Omitting proofs}\label{supsec:Omitting proofs}

\begin{proofof}{Proof of Lemma~\ref{lem:xi}}
Fix $t$ and let us consider the formula that produces the coordinates of $\hat{x}^t$, given $x^{t-1}, y^{t-1}$. For simplicity in writing, we drop the superscript $t-1$ and refer to $x, y$ as the strategies of the two players computed at the end of iteration $t-1$. Focusing on the row player (the same argument follows for the column player too), we know that for every $i\in [n]$,
\begin{equation*}
\hat{x}_i^t  = x_i \cdot \frac{e^{\xi e_i^T R y} }{\sum\nolimits_{j = 1}^{n} x_j e^{\xi e_j^T R y}} .
\end{equation*}

We want to compute for every $i\in [n]$, the limit $\lim_{\xi\to\infty} \hat{x}_i^t$.
We distinguish two cases, depending on whether $e_i$ is a best-response strategy against $y$ or not. Denote by $B(y)$ the set of the pure best-response strategies of the row player against $y$, i.e., $B(y) = \{i: e_i \mbox{ is a best response against } $y$ \}$.

We start with the first case, where $i\in B(y)$. This means that $e_i^TRy$ is a best-response payoff against $y$ and therefore $e^{\xi e_i^T R y} = e^{\xi e_j^T R y}$ for any $j\in B(y)$. Based on this, we can now write the limit as 
\begin{equation*}
    lim_{\xi\to\infty} \hat{x}_i^t=x_i \cdot \frac{e^{\xi e_i^T R y} }{\sum\nolimits_{j \in B(y)} x_j e^{\xi e_j^T R y} + \sum\nolimits_{j \notin B(y)} x_j e^{\xi e_j^T R y}}=x_i \cdot \frac{1 }{\sum\nolimits_{j \in B(y)} x_j  + \sum\nolimits_{j \notin B(y)} x_j e^{\xi (e_j^T R y- e_i^T Ry)}},
\end{equation*}

Obviously, we can see that $(e_j^T R y - e_i^TR y)< 0$ for any $j \notin B(y)$. Hence, as $\xi$ goes to infinity, it holds that $\lim_{\xi\to\infty} = \sum\nolimits_{j \notin B(y)} x_j e^{\xi (e_j^T R y- e_i^T Ry)} = 0$. Thus, we have that 

\begin{equation*}
    \lim_{\xi\to\infty} \hat{x}_i^t =   \frac{x_i}{\sum\nolimits_{j\in B(y)} x_j }. 
\end{equation*}

Consider now the second case where $i\not\in B(y)$. Let 
$A(y)$ be the set of indices $j$, such that $e^{\xi e_j^TRy} > e^{\xi e_i^TRy}$. Since $i\not\in B(y)$, we know that $A(y)\neq\emptyset$. 
Then, we have 
\begin{equation*}
 \lim_{\xi\to\infty} \hat{x}_i^t = x_i \cdot \frac{e^{\xi e_i^T R y} }{\sum\nolimits_{j \in A(y)} x_j e^{\xi e_j^T R y} + \sum\nolimits_{j \notin A(y)} x_j e^{\xi e_j^T R y}}=x_i \cdot \frac{1 }{\sum\nolimits_{j \in A(y)} x_j e^{\xi (e_j^T R y- e_i^T Ry)}  + \sum\nolimits_{j \notin A(y)} x_j e^{\xi (e_j^T R y- e_i^T Ry)}}. 
\end{equation*}
As $\xi$ goes to infinity, it holds that $\lim_{\xi\to\infty} = \sum\nolimits_{j \in A(y)} x_j e^{\xi (e_j^T R y- e_i^T Ry)} = +\infty$. 
For the second term of the denominator, we can see that for the indices $j\not\in A(y)$, for which $e^{\xi e_j^TRy} < e^{\xi e_i^TRy}$, it holds that  $\lim_{\xi\to\infty} = \sum\nolimits_{j} x_j e^{\xi (e_j^T R y- e_i^T Ry)} = 0$. For the remaining indices $j$, for which $e^{\xi e_j^TRy} = e^{\xi e_i^TRy}$, the limit is $\sum_j x_j \leq 1$. Thus, in total, we have that 
\begin{equation*}
    \lim_{\xi\to\infty} \hat{x}_i^t = 0.
\end{equation*}
Hence, when $\xi\to\infty$, the strategy $\hat{x}$ will eventually contain in its support only best responses to $y$, and therefore $\hat{x}$ will form a best response to $y$ as well.
\end{proofof}


\begin{proofof}{Proof of Theorem~\ref{thm:neighborhood}}
The proof is based on the following lemma, shown in \cite{EY10}, which we state here for the case of zero-sum games:
\setcounter{lemma}{1}
\begin{lemma}
\label{lem:EY}
Consider a zero-sum game given by matrix $R$ with a unique Nash equilibrium $(x^*, y^*)$, and let $|R|$ be the number of bits needed for the representation of $R$. There exists a polynomial $p$ such that for every $\delta>0$, every $\varepsilon$-Nash equilibrium $(x, y)$ satisfies that $|x^*_i - x_i| < \delta$, as long as $\varepsilon \leq 1/2^{p(|R|+size(\delta))}$, where $size(\delta) = O(\log(1/\delta))$ is the number of bits needed for representing $\delta$.
\end{lemma}

By the assumptions in the statement of Theorem~\ref{thm:neighborhood}, we fix $\varepsilon = c\cdot \eta^{1/\rho}$, for some constant $c$, so that $(x^t, y^t)$ is an $\varepsilon$-Nash equilibrium. We claim  that there exists $\delta(\eta)$ such that $\varepsilon$ and $\delta(\eta)$ satisfy the inequality stated in Lemma \ref{lem:EY}. In particular, by looking more carefully at the desired inequality and solving with respect to $\delta$, one can construct a function $\delta(\eta)$, such that for the given $\varepsilon$ we have selected, it holds that
\begin{equation*}
    \varepsilon \leq 1/2^{p(|R|+size(\delta(\eta)))} \mbox{  and  } \lim\nolimits_{\eta\rightarrow 0} \delta(\eta)= 0.
\end{equation*}

Hence, we can now apply Lemma \ref{lem:EY} and obtain that for any $\varepsilon$-Nash equilibrium $(x, y)$ we have that $|x^*_i - x_i| \leq \delta(\eta)$ and $|y^*_i - y_i| \leq \delta(\eta)$.   
The proof now of Theorem~\ref{thm:neighborhood} is immediate, since 
$||(x^*, y^*) - (x^t, y^t)||_1 = \sum_{i=1}^n|x^*_i - x_i| + \sum_{i=1}^n|y^*_i - y_i| \leq 2n\cdot \delta(\eta)$, which goes to $0$ as $\eta\rightarrow 0$. 
\end{proofof}


\begin{proofof}{Proof of Theorem~\ref{thm:contraction}}
To prove the theorem, we describe first a discrete dynamical system that captures the FLBR-MWU dynamics, and we will prove that for an appropriate norm of the Jacobian matrix of the system, its value is less than one\footnote{Besides \cite{Galor}, readers could advise Chapter 7~\cite{QuartSaccoSaleri}.}. The update rule $\varphi$ of FLBR-MWU is
\begin{equation}
\begin{split}
        & \varphi(x, y)  = (\varphi_1(x, {y}), \varphi_2({x}, {y})), \mbox{ where } \qquad \qquad \qquad \\
        & \varphi_{1,i}({x}, {y}) = (\varphi_1({x}, {y}))_i = x_i \tfrac{e^{\eta e^T_i R f({x}, {y})}}{\sum_\ell x_\ell e^{\eta e^T_\ell R f({x}, {y})}}, \\
        & \varphi_{2,i}({x}, {y}) = (\varphi_2({x}, {y}))_i = y_i \tfrac{e^{-\eta e^T_i R^T h({x}, {y})}}{\sum_\ell y_\ell e^{-\eta e^T_\ell R^T h({x}, {y})}}, 
\end{split}
\end{equation}
where $f(x,y)$ and $h(x,y)$ are column vectors with
$\displaystyle (f({x}, {y}))_i = y_i \tfrac{e^{-\xi e^T_iR^T{x} }}{\sum_\ell y_\ell e^{ -\xi e^T_\ell R^T{x} }}$, and $(h({x}, {y}))_i = x_i \tfrac{e^{\xi e^T_iR{y} }}{\sum_\ell x_\ell e^{\xi e^T_\ell R {y} }}$, for all $i\in \{1,\dots,n\}$.

\noindent Clearly, the dynamics of FLBR-MWU are captured by $({x}^{t+1}, {y}^{t+1}) = \varphi({x}^t, {y}^t)$. The Jacobian of $\varphi$ is a $2n\times 2n$ matrix, which can be written in the form of a $2 \times 2$ block matrix, as follows:

\begin{equation}
\label{eq:jacobian}
    J = \left(\begin{array}{cc}
         \frac{\partial \varphi_1}{ \partial \mathbf{x}} &  \frac{\partial \varphi_1}{ \partial \mathbf{y}} \\
        \frac{\partial \varphi_2}{ \partial \mathbf{x}} &  \frac{\partial \varphi_2}{ \partial \mathbf{y}}
    \end{array}\right).
\end{equation}

\noindent In order to use Theorem \ref{thm:App_jac_contr_thm} and prove that $\phi$ is a contraction, we need to argue about the eigenvalues of $J$ at the equilibrium $({x^*}, {y^*})$. Towards this, in Subsection \ref{supplementary:sub:App_Jac}, we provide the exact form of each entry of $J$ at $(x^*,y^*)$ (after some simplification steps by exploiting the fact that $(x^*,y^*)$ is an equilibrium). 

\noindent We analyze first the eigenvalues that are derived by the rows of $J$ that correspond to $\varphi_{1, i}$ for some $i \not \in supp({x}^*)$ and to $\varphi_{2, i}$ for some $i \not \in supp({y}^*)$.
Let ${x}^{*T} R {y}^* = v$ be the value of the game.
By referring to Subsection \ref{supplementary:sub:App_Jac}, we have that for any $i \not \in supp({x}^*)$:
\begin{equation*}
\frac{\partial \varphi_{1,i}}{\partial x_i}(x^*,y^*) = \frac{e^{\eta e^T_i R y^*}}{e^{\eta v}}, \enspace \frac{\partial \varphi_{1,i}}{\partial x_j}(x^*,y^*) = 0 \enspace \text{for any } i \neq j, \text{ and } \frac{\partial \varphi_{1,i}}{\partial y_j}(x^*,y^*) = 0 , \enspace \text{for any }j.
\end{equation*}
Hence, the $i$-th row of the upper block of $J$ has only one non-zero entry, namely, the diagonal element, provided that $i \not \in supp({x}^*)$. Thus, $\tfrac{e^{\eta e^T_i R y^*}}{e^{\eta v}}$ is an eigenvalue of $J$ at $({x}^*, {y}^*)$. 
We note also that\footnote{A unique Nash equilibrium of a zero-sum game is also a quasi-strict equilibrium (Theorem 1 in \cite{DBLP:journals/mp/Norde99}), meaning that strategies that are not in the support of the equilibrium have strictly less payoff than the best-response payoff.} $e_i^TR {y}^*< v  $ for $i \not \in supp(x^*)$, hence $|\frac{\partial \varphi_{1,i}}{\partial x_i}(x^*,y^*)| < 1$. Analogously, for $i \not \in supp({y}^*)$ we have that $\frac{\partial \varphi_{2,i}}{\partial y_i}(x^*,y^*) = \tfrac{e^{- \eta e^T_i R^T x^*}}{e^{- \eta v}}$, whereas all other partial derivatives of $\varphi_{2,i}$ are zero. Thus, $\tfrac{e^{- \eta e^T_i R^T x^*}}{e^{- \eta v}}$ is also an eigenvalue of $J$, with $|\tfrac{e^{- \eta e^T_i R^T x^*}}{e^{- \eta v}}| < 1$, since $e^T_i R^T x^* > v$ for $i\notin supp(y^*)$ by footnote 2.

\noindent We now focus on the rows and columns that correspond to the support of ${x}^*$ and ${y}^*$. We denote this submatrix as $\tilde{J}$, with $k_1 = |supp({x}^*)|$, $k_2 = |supp({y}^*)|$ and $k = k_1 + k_2$. Thus, $\tilde{J} \in \mathbb{R}^{k \times k}$. It can been seen that $J$ has eigenvalues with absolute value less that one iff the same holds for $\tilde{J}$ as well.

\noindent Using equations~$\eqref{eq:Jac_xy_at_equilibrium}$ and computing $((\boldsymbol{1}_{k_1}, \boldsymbol{0}_{k_2})^T\cdot \tilde{J})_j$ for an arbitrary coordinate $j$, we end up with the quantity $\sum\nolimits_i x^*_i\sum\nolimits_k R_{ik} y^*_k R^T_{kj} - \sum\nolimits_i x^*_i \sum\nolimits_k x^*_k\sum\nolimits_l R_{kl} y^*_l R^T_{lj}$, that equals zero. Thus, $(\boldsymbol{1}_{k_1}, \boldsymbol{0}_{k_2})$ is a left eigenvector of $\tilde{J}$ corresponding to the zero eigenvalue. Using the same argumentation we have that $(\mathbf{0}_{k_1}, \mathbf{1}_{k_2})$ is also a left eigevector of $\tilde{J}$ with eigenvalue zero.

\noindent We will make use of the following claim, regarding orthogonal pairs of eigenvectors.
\begin{claim}\label{claim:left_right_eigenvectors}
Consider a matrix $A \in \mathbb{R}^{n \times n}$, 
an eigenvalue $\lambda$ and a left eigenvector $u^T$, corresponding to $\lambda$. Then for every right eigenvector $v$ that does not correspond to $\lambda$, it holds that $u^T v = 0$. 
\end{claim}

\noindent The proof of the claim, which is a simple linear algebra exercise, is at the end of this section. 
From Claim~\ref{claim:left_right_eigenvectors}, it follows that for any right eigenvector $(\tilde{x}, \tilde{y})$ corresponding to a nonzero eigenvalue, we have 
\begin{equation}
\label{eq:right-eigen}
    \tilde{\mathbf{x}}^T \mathbf{1}_{k_1} = 0 \mbox{ and } \tilde{\mathbf{y}}^T \mathbf{1}_{k_2} = 0. 
\end{equation}

\noindent With that in hand, let us now rewrite $\tilde{J}$, as $\tilde{J} = J' + A$, where $J'$ is produced by 
deleting the term $-x^*_i$ (resp. $-y^*_i$) from every element of the upper left (resp. lower right) block of $\tilde{J}$. I.e., $A$ contains $-x^*_i$ in all entries of the $i$-th row in the upper left block, and $-y^*_i$ in all entries of the $i$-th row in the bottom right block. The other two blocks of $A$ contain only zeros.
Using \eqref{eq:right-eigen}, we can see that for every non-zero eigenvalue $\lambda$ of $\tilde{J}$, that corresponds to a right eigenvector $(\tilde{x}, \tilde{y})$, it holds that $A\cdot (\tilde{x}, \tilde{y})= 0$, thus $\lambda$ is also an eigenvalue of the matrix $J'$. By the equations in Subsection \ref{supplementary:sub:App_Jac}, we can write $J'$ as a $2 \times 2$ block matrix, as follows.
\begin{equation*}
    J' = \left(
    \begin{array}{cc}
        I_{k_1 \times k_1} + \eta \xi D^{xx} & \eta D^{xy} \\
        \eta D^{yx} & I_{k_2 \times k_2} + \eta \xi D^{yy}
    \end{array}
    \right),
\end{equation*}
with 
\begin{itemize}
\item[] $D^{xx}_{ij} = -x^*_i \left(\sum\nolimits_k R_{ik} y^*_k R^T_{kj} - \sum\nolimits_k x^*_k \sum\nolimits_l R_{kl} y^*_l R^T_{lj}\right)$, with $i,j \in [k_1] $,
\item[] $D^{yy}_{ij} = -y^*_i \Big(\sum\nolimits_k R_{kj}R^T_{ik} x^*_k - \sum\nolimits_k y^*_k \sum\nolimits_l R^T_{kl} x^*_l R_{lj}\Big)$, with $i,j \in [k_2] $, 
\item[] $D^{yx}_{ij} = -y^*_i \Big (R^T_{ij}-e_j^TRy^*\Big) \tfrac{e^{\xi e^T_j R y^* }}{e^{\xi v}}$, with $i \in [k_2]$, $ j \in [k_1] $,
\item[] $D^{xy}_{ij} = x^*_i \Big(R_{ij} - e_j^TR^T x^*\Big) \tfrac{e^{-\xi e^T_j R^T x^* }}{e^{-\xi v}}$, with $i \in [k_1] $, $j \in [k_2] $.
\end{itemize}

\noindent We observe that all the entries of the matrices $D^{xx}, D^{yy}, D^{yx}, D^{xy} $ are within the interval $[-1, 1]$.
Furthermore, for the remainder of the proof, and without loss of generality, we assume that $(x^*, y^*)$ is a mixed strategy profile, i.e., both $x^*$ and $y^*$ are mixed.\footnote{If exactly one of $x^*$, $y^*$ were mixed, this would also imply the existence of a pure equilibrium, contradicting our uniqueness assumption. If $(x^*, y^*)$ is a pure strategy profile, then $\tilde{J}$ is a $2 \times 2$ block matrix, where each block is a single element. Using the equations of $\eqref{eq:Jac_xy_at_equilibrium}$ in Subsection \ref{supplementary:sub:App_Jac}, then the matrix $\tilde{J}$ only has the zero eigenvalue and the proof of Theorem \ref{thm:contraction} follows directly.}

\noindent We now consider the diagonal element of $D^{xx}$, for any $i$, that is,
\begin{equation*}
    -x^*_i\Big(\sum\nolimits_l R^2_{il} y^*_l - \sum\nolimits_k x^*_k \sum\nolimits_l R_{kl} y^*_l R^T_{li}\Big)
\end{equation*}
We establish the following useful property.

\begin{lemma}
\label{lem:Dxx-negative}
For any $i\in [k_1]$, $D^{xx}_{ii} < 0$, and for any $j\in [k_2]$, $D^{yy}_{jj} < 0$. 
\end{lemma}

\begin{proofof}{Proof of Lemma \ref{lem:Dxx-negative}}
We first prove that for any $i$, $D^{xx}_{ii} \leq  0$.
For the sake of contradiction, assume that there exists an index $i$, such that $D^{xx}_{ii} > 0$. 
This means that 
\begin{equation*}
    \sum\nolimits_l R^2_{il} y^*_l <\sum_kx^*_k\sum_l  R_{kl}y^*_lR_{il}.
\end{equation*}
To proceed, we claim that
\begin{equation}
\label{eq:contradiction}
    v\leq\sum\nolimits_l R_{il} z_l,
\end{equation}
where $z_l = \frac{R_{il}y^*_l}{v}$, and $z = (z_l)_{l\in [n]}$. 
To see this, it is crucial to notice first that both $y^*$ and $z$ are probability vectors and also that $v = \sum\nolimits_l R_{il} y_l^*$. 
Hence, the LHS  and the RHS of Equation \eqref{eq:contradiction} are two different convex combinations of the $R_{il}$ values. To go from the LHS to the RHS, we simply replace $y_l^*$ by $z_l$. For each $R_{il}$ that is itself less than $v$, the coefficient $y_l^*$ is replaced by a smaller coefficient, since $z_l < y_l^*$ in this case (by the definition of $z_l$). On the contrary, for each $R_{il}$ with $R_{il} > v$, it holds that $z_l> y_l^*$ (and we also have $z_l = y_l^*$ when $R_{il} = v$). Hence, we can think of the move from the LHS to the RHS of \eqref{eq:contradiction}, as transferring probability mass from the lowest valued $R_{il}$'s to the highest ones. Let $\Delta$ be the total amount of probability mass that was transferred. Then $\Delta = \sum_{l:R_{il < v}} (y_l^* - z_l) \geq 0$. Note that it also holds that $\Delta = \sum_{l:R_{il > v}} (z_l - y_l^*)$. If we compare now the LHS with the RHS, the RHS has a deficit of a total value of at most $\Delta\cdot v$ from the terms with $R_{il} < v$, compared to the corresponding terms of the LHS. At the same time, it has a surplus of at least $\Delta\cdot v$ from the terms with $R_{il} > v$. Combining the deficit and the surplus, this proves Equation \eqref{eq:contradiction}.

Using \eqref{eq:contradiction}, we can now obtain the following contradiction:
\begin{equation*}
    v\leq\sum\nolimits_l R_{il} z_l <  x^{*T}Rz = v,
\end{equation*}
where the strict inequality above follows by the condition stated just before Equation \eqref{eq:contradiction}, and the final equality holds since $x^{*T}Re_j = v$ for any $j\in supp(y^*)$ (and so for any $j\in supp(z)$). 

Thus, we have reached a contradiction, which means that $D^{xx}_{ii} \leq 0$ for every $i\in [k_1]$. 
In addition, it is not difficult to see that in case $D^{xx}_{ii} = 0$ for some $i$, the strategy profile $(i, y^*)$ is also a Nash equilibrium. But this would imply that there also exists a pure equilibrium formed by $i$ and a pure best response (from the support of $y^*$), contradicting the fact that we have a unique equilibrium. Hence, $D^{xx}_{ii}$ is strictly negative for every $i \in [k_1]$. 

Similarly, the same analysis holds for the matrix $D^{yy}$, completing the proof of the lemma. 
\end{proofof}

\noindent To finish the proof, we estimate an upper bound on the $p$-norm of $J'$ for $p\in\mathbb{N}$. We have that 

\begin{tabbing}
$\|J'\|^p_{p}$ \= $ = \sum_j \Big(\sum_i |J'_{ij}|^p\Big) $ \= $ \leq k \max_j \Big(\sum_i |J'_{ij}|^p\Big)$  \\
\>\> $ \leq k( |1 + \eta \xi D^{xx}_{j'j'}|^p $ + $ \eta^p \xi^p  \sum_{\substack{i=1, \\ i \neq j}}^{k_1} |D^{xx}_{ij'}|^p + \eta^p \sum_{i = k_1 + 1}^{k_2} |D^{yx}_{ij'}|^p)$ \\
\>\> $\leq k(|1 + \eta \xi D^{xx}_{j'j'}|^p + \eta^p \xi^p k_1 + \eta^p k_2)$,\\
\end{tabbing}
where $\ell$ is the column of $J'$ that achieves the maximum sum, i.e., $\ell = \argmax_{j \in [k_1 + k_2]} \sum_i |J'_{ij}|^p$, and we assumed without loss of generality that $\ell$ belongs to $\{1,\dots,k_1\}$. We can now see that since $D^{xx}_{\ell\ell}$ is negative, and both $|D^{xx}_{i\ell}|$, $|D^{yx}_{i\ell}|$ are at most equal to one, then if $\eta \xi<1$, and $\eta$ is sufficiently small, there exists an appropriate $p$ so that $\|J'\|^p_{p} < 1$. However, it is well known that the maximum absolute value of an eigenvalue of a matrix is bounded by the induced matrix norms, therefore is suffices to check that $\|J'\| < 1$ for some matrix norm, see \cite{QuartSaccoSaleri}. Thus, the absolute value of the maximum eigenvalue of $J'$ is less than one, and this concludes our proof.
\end{proofof}

\begin{proofof}{Proof of Claim \ref{claim:left_right_eigenvectors}}
Consider two distinct eigenvalues of $A$ $\lambda_1$ and $\lambda_2$, such that $v$ is the corresponding to $\lambda_1$ left eigenvector, while $u$ is the corresponding to $\lambda_2$ right eigenvector (\cite{strang09}). In other words, $v$ is the corresponding to $\lambda_1$ right eigenvector for $A^T$. We observe that, $v^T (A^T u) = (v^T A^T) u = (Av)^T u$. So, $\lambda_1 v^T u = (A^T v)^T u =  v^T (A u) = v^T \lambda_2 u = \lambda_2 v^T u$. Thus, $v^T u = 0$.
\end{proofof}


\subsection{Equations of the Jacobian Entries}\label{supplementary:sub:App_Jac}

Recall the form of the Jacobian of our dynamical system in Equation \eqref{eq:jacobian}.

We compute the form of each entry of $J$ at the point $({x}, {y})$. Let $Q_x = \sum\nolimits_\ell x_\ell e^{\eta e^T_\ell R f({x}, {y})}$, $Q_y = \sum\nolimits_\ell y_\ell e^{ - \eta e^T_\ell R^T h({x}, {y})}$, $S_x = \sum\nolimits_\ell x_\ell e^{\xi e^T_\ell R {y} }$, and $S_y = \sum\nolimits_\ell y_\ell e^{- \xi e^T_\ell R^T{x} }$. 

\begin{equation}\label{eq:Jac_xy}
    \begin{split}
        & \tfrac{\partial \varphi_{1,i}}{\partial x_i} =  e^{\eta e^T_i R f({x}, {y})} \tfrac{Q_x \left(1 + \eta x_i \tfrac{\partial}{\partial x_i} (e_i^TR f({x}, {y}))\right) - x_i \tfrac{\partial}{\partial x_i} Q_x}{Q^2_x}, \quad i \in [n], \\
        & \tfrac{\partial \varphi_{1,i}}{\partial x_j} =  x_i e^{\eta e^T_i R f({x}, {y})}\tfrac{ \eta  Q_x \tfrac{\partial}{\partial x_j} (e_i^TR f({x}, {y})) - \tfrac{\partial}{\partial x_j} Q_x}{Q^2_x}, \quad i,j \in [n] \text{ and } i \neq j, \\
        & \tfrac{\partial \varphi_{1,i}}{\partial y_j} =  x_i e^{\eta e^T_i R f({x}, {y})}\tfrac{ \eta Q_x\tfrac{\partial}{\partial y_j} (e_i^TR f({x}, {y})) - \tfrac{\partial}{\partial y_j} Q_x}{Q^2_x}, \quad i,j \in [n], \\
        & \tfrac{\partial \varphi_{2,i}}{\partial x_j} =  y_i e^{-\eta e^T_i R^T h({x}, {y})} \tfrac{ -\eta Q_y \tfrac{\partial}{\partial x_j} (e_i^TR^T h({x}, {y})) - \tfrac{\partial}{\partial x_j} Q_y}{Q^2_y}, \quad i,j \in [n],\\
        & \tfrac{\partial \varphi_{2,i}}{\partial y_i} =  e^{-\eta e^T_i R^T h({x}, {y})} \tfrac{Q_y \left(1 - \eta y_i \tfrac{\partial}{\partial y_i} (e_i^TR^T h({x}, {y})) \right) - y_i \tfrac{\partial}{\partial y_i} Q_y}{Q^2_y}, \quad i \in [n], \\
        & \tfrac{\partial \varphi_{2,i}}{\partial y_j} =  y_i e^{- \eta e^T_i R^T h({x}, {y})} \tfrac{ - \eta Q_y \tfrac{\partial}{\partial y_j} (e_i^TR^T h({x}, {y})) - \tfrac{\partial}{\partial y_j} Q_y}{Q^2_y}, \quad i,j \in [n] \text{ and } i \neq j. \\
    \end{split}
\end{equation}

\noindent At the point $({x^*}, {y}^*)$, after exploiting the fact that this is an equilibrium profile, and simplifying some of the calculations, we obtain the following forms. 

\begin{equation}\label{eq:Jac_xy_at_equilibrium}
    \begin{split}
        & \tfrac{\partial \varphi_{1,i}}{\partial x_i} =  1 - x^*_i \left(\eta \xi \left( \sum\nolimits_k R^2_{ik} y^*_k - \sum\nolimits_k x^{*}_k \sum\nolimits_l R_{kl} y^{*}_l R^T_{li} \right) + 1\right), \enspace i \in supp({x}^*), \\
        & \tfrac{\partial \varphi_{1,i}}{\partial x_i} = \tfrac{e^{\eta e^T_i R {y}^*}}{e^{\eta v}}, \enspace i \not \in supp({x}^*), \\
        & \tfrac{\partial \varphi_{1,i}}{\partial x_j} =  - x^*_i \left( \eta \xi \left( \sum\nolimits_k R_{ik} y^*_k R^T_{kj} - \sum\nolimits_k x^{*}_k \sum\nolimits_l R_{kl} y^{*}_l R^T_{lj} \right) + 1 \right), \enspace i \in supp({x}^*), i\neq j, \\
        & \tfrac{\partial \varphi_{1,i}}{\partial x_j} =  0,  \enspace i \not \in supp({x}^*) \text{ and } i \neq j, \\
        & \tfrac{\partial \varphi_{1,i}}{\partial y_j} =  x^*_i \eta (R_{ij} - e_j^TR^T x^*)\tfrac{e^{-\xi e^T_j R^T x^* }}{e^{-\xi v}}, \enspace \text{for all } i \in supp({x}^*), \\
        & \tfrac{\partial \varphi_{1,i}}{\partial y_j} =  0,  \enspace i \not \in supp({x}^*), \\
        & \tfrac{\partial \varphi_{2,i}}{\partial x_j} =  - y^*_i \eta (R^T_{ij} - e_j^TR y^*)\tfrac{e^{\xi e^T_j R y^* }}{e^{\xi v}}, \enspace \text{for all } i \in supp({y}^*), \\
        & \tfrac{\partial \varphi_{2,i}}{\partial x_j} =  0,  \enspace i \not \in supp({y}^*), \\
        & \tfrac{\partial \varphi_{2,i}}{\partial y_i} =  1 - y^*_i \left( \eta \xi \left( \sum\nolimits_k (R^T_{ik})^2 x^*_k - \sum\nolimits_k y^*_k \sum\nolimits_l R^T_{kl} x^*_l R_{li}\right) + 1\right), \enspace i \in supp({y}^*), \\
        & \tfrac{\partial \varphi_{2,i}}{\partial y_i} =  \tfrac{e^{-\eta e^T_i R^T {x}^*}}{e^{-\eta v}},  \enspace i \not \in supp({y}^*), \\
        & \tfrac{\partial \varphi_{2,i}}{\partial y_j} =  - y^*_i \Big(\eta \xi \Big( \sum\nolimits_k R_{kj} R^T_{ik} x^*_k - \sum_k y^*_k \sum_l R^T_{kl} x^*_l R_{lj} \Big) + 1\Big), \enspace i \in supp({y}^*), i\neq j, \\
        & \tfrac{\partial \varphi_{2,i}}{\partial y_j} =  0,  \enspace i \not \in supp({y}^*) \text{ and } i \neq j, \\
    \end{split}
\end{equation}


\section{Additional Numerical Demonstrations}\label{supsec:Numerical Experiments}

In this section, we demonstrate the properties of the FLBR-MWU algorithm using additional metrics and perform further comparisons. 

\noindent {\bf Convergence to the value of the game.} Figure \ref{game:value:fig} shows the evolution of the current value of the game at each iteration, with the same payoff matrix as that used in the example of Figure~\ref{motiv:ex:fig}, in the main body of the paper. The current value of the game at iteration $t$ is defined as $v^t = (x^{t})^{T} R y^t$, and it serves as another convergence measure to Nash equilibrium. MWU (blue) oscillates around the true value of the game ($v = 0.529677$) without converging, while OMWU (red) oscillates with decreasing amplitude and eventually it converges to the true value. The current game value for the FLBR-MWU dynamics (black) converges much faster requiring only a few thousand steps.

\begin{figure}[htp]
    \centering
    \includegraphics[width=.8\textwidth]{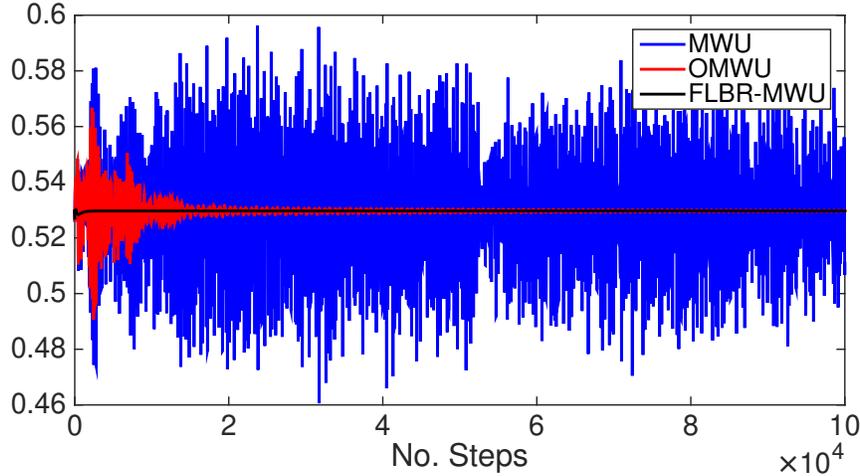}
    \caption{The value of the game as a function of the number of steps for the three MWU variants.}
    \label{game:value:fig}
\end{figure}

\begin{figure}[htp]
    \centering
    \includegraphics[width=1.05\textwidth]{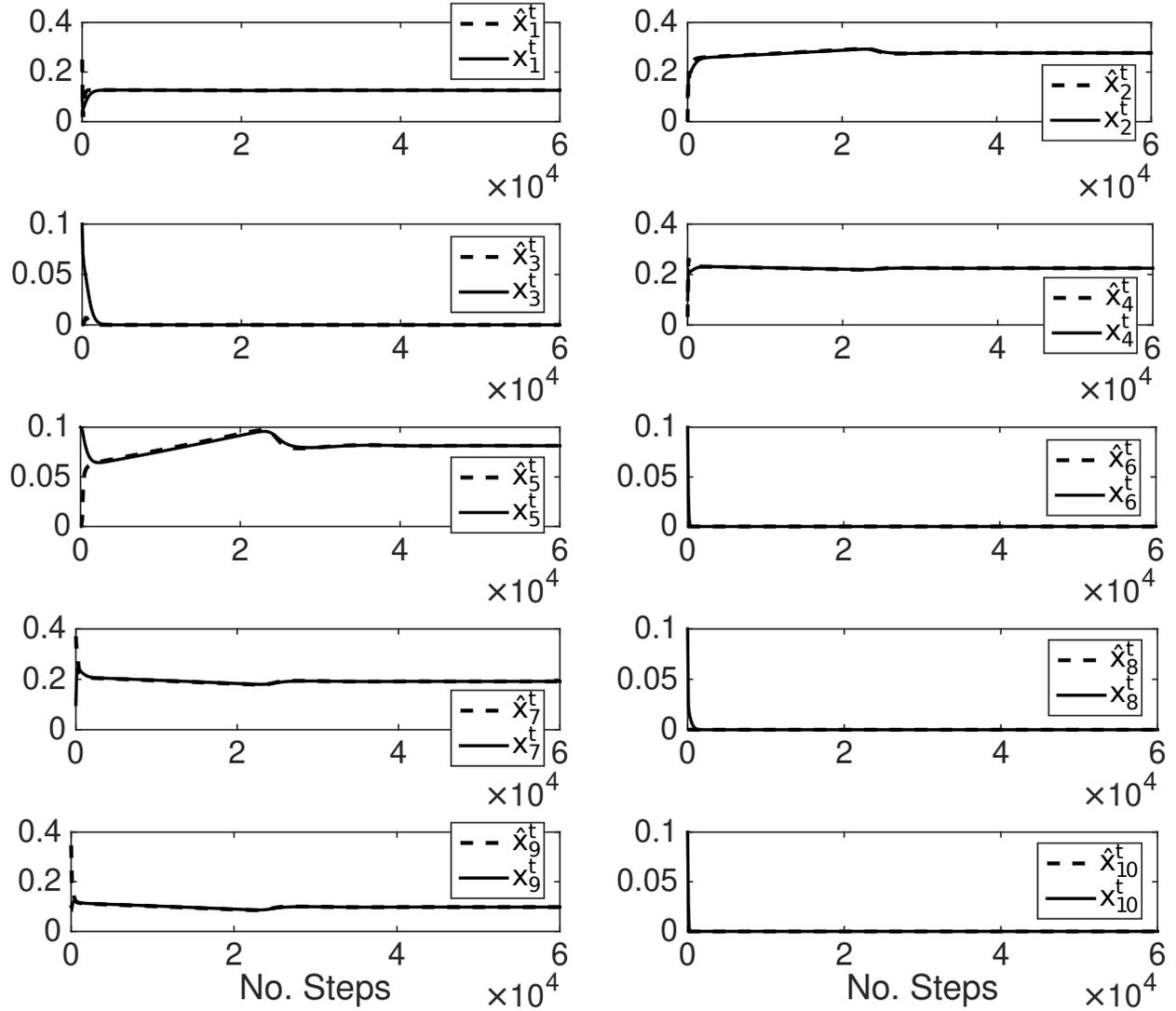}
    \caption{The dynamics of the update step per coordinate (solid), as well as the IBR step (dashed) for the row player. For the equilibrium strategy $x^*$, it holds that $supp(x^*)=\{1,2,4,5,7,9\}$. Note that $x^t$ converges to the same support.}
    \label{x_t:linear:fig}
\end{figure}

\begin{figure}[htp]
    \centering
    \includegraphics[width=1.05\textwidth]{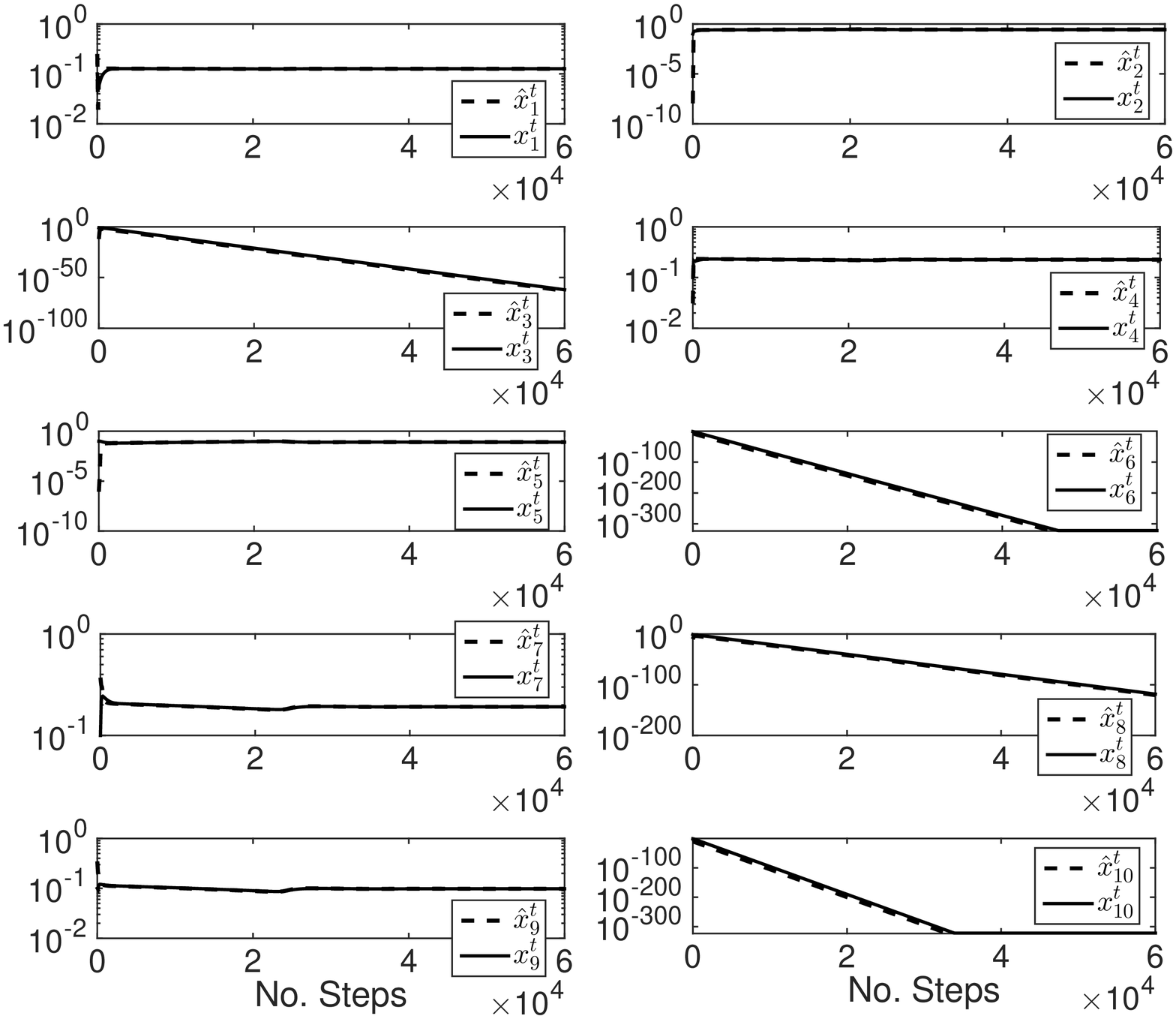}
    \caption{Same as Figure \ref{x_t:linear:fig}, but in logarithmic scale.}
    \label{x_t:log:fig}
\end{figure}

\noindent {\bf Dynamics trajectories.} Figures \ref{x_t:linear:fig} and \ref{x_t:log:fig} show the trajectories of the row player (i.e., $x_i^t$ for $i=1,\dots, 10$) in linear and log scale, respectively. Similarly, Figures \ref{y_t:linear:fig} and \ref{y_t:log:fig} show the trajectories of the column player (i.e., $y_i^t$). Again, the payoff matrix is the same as in Figure~\ref{motiv:ex:fig} from the main paper, and the Nash equilibrium is estimated as:
\begin{equation*}
    \left(\begin{array}{c}
         x^{*T}  \\
         y^{*T}
    \end{array}\right) = \left(\begin{array}{l}
        0.126766, ~0.276988, ~0, ~0.22506, ~0.081435, ~0, ~0.191705, ~0, ~0.098045, ~0  \\
        0, ~0.058227, ~0, ~0.298188, ~0.213176, ~0, ~0, ~0.283403, ~0.000376, ~0.146628 
    \end{array}\right) \ .
\end{equation*}

First, we note that for all pure strategies that do not belong to the support of $x^*$ or $y^*$, the corresponding probabilities in $x^t$ and $y^t$ converge to 0 under FLBR-MWU, after a few thousand steps. Additionally, we observe interesting patterns during the evolution of the learning dynamics in both scales which are intimately connected with the KL divergence trajectory shown in Figure~\ref{motiv:ex:fig} of the main paper. Indeed, it is worth looking at $y_9^t$ (log scale; Figure \ref{y_t:log:fig}), which shows the most interesting pattern. Initially it seems that this is not a surviving strategy of the dynamics and its probability decreases for the first $10K$ steps. However, and, despite its very low value, it recovers to the actual Nash equilibrium value. Similarly, we observe that the non-zero elements of $x^t$ (linear scale; Figure \ref{x_t:linear:fig}) are linearly evolving for several thousands of steps. Those changes in the dynamics correspond to the plateau of the KL divergence observed in Figure~\ref{motiv:ex:fig} of the main text. Our explanation of the dynamics trajectories is as follows: starting from the uniform state, the FLBR-MWU algorithm first finds an approximate Nash equilibrium with a value close to the true value of the game but then escapes from it until it eventually converges to the actual Nash equilibrium.

Another interesting observation is that the dynamics of the IBR step (recall Equation~$\eqref{step2}$ in the main paper) drive the FLBR-MWU dynamics in the sense that when the IBR dynamics are above the FLBR-MWU dynamics, then the corresponding probabilities in the update step of FLBR-MWU increase, while the opposite is true when the IBR dynamics are below the FLBR-MWU dynamics.

\newpage

\begin{figure}[htp]
    \centering
    \includegraphics[width=\textwidth]{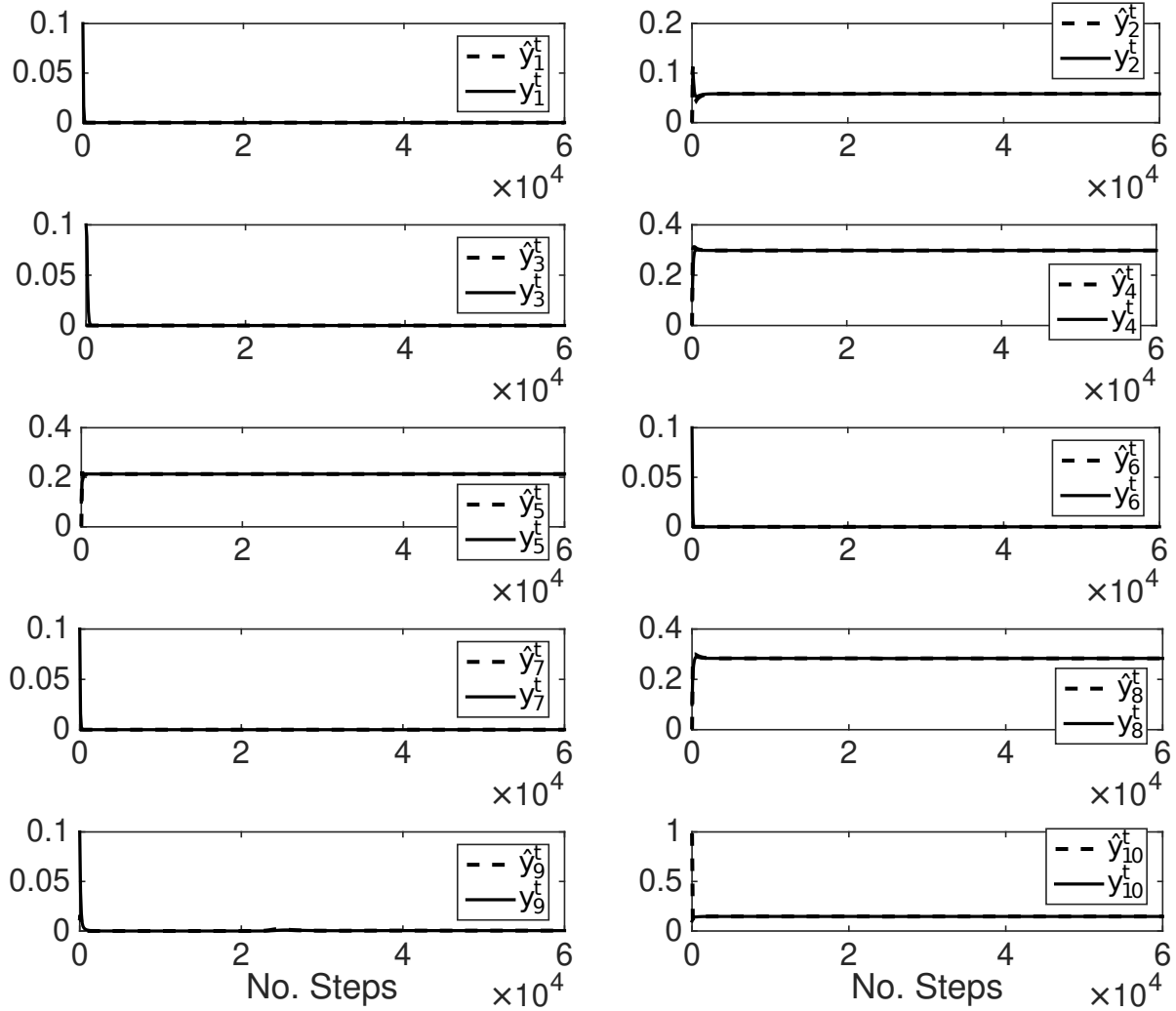}
    \caption{The dynamics of the update step per coordinate (solid), as well as the IBR step (dashed) for the column player. For the equilibrium strategy $y^*$, it holds that $supp(y^*)=\{2,4,5,8,9,10\}$.}
    \label{y_t:linear:fig}
\end{figure}

\begin{figure}[htp]
    \centering
    \includegraphics[width=\textwidth]{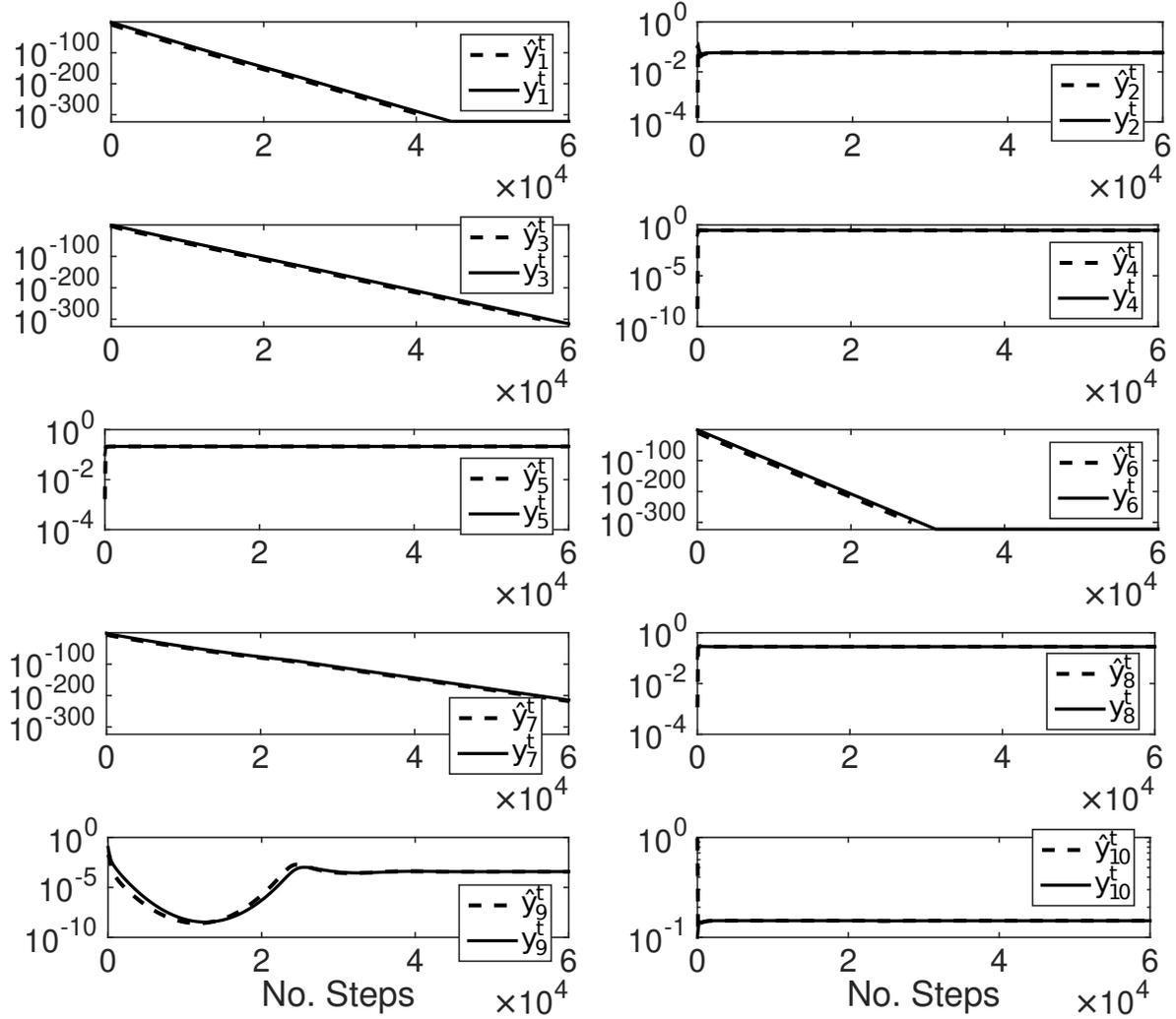}
    \caption{Same as Figure \ref{y_t:linear:fig}, but in logarithmic scale.}
    \label{y_t:log:fig}
\end{figure}

\noindent {\bf Effect of the intermediate rate ($\xi$).} We present further statistical information on the effect of $\xi$. Figure~\ref{ksi:effect:fig} shows the distribution of the number of steps as a boxplot for $n=10$ (left) and $n=20$ (right). The red line in the boxplot corresponds to the median value while the blue box corresponds to the area covered by the 2nd and 3rd quantiles. The distribution of the number of steps till convergence is positively (or right) skewed. Therefore we also report the statistics of the right tail in Table~\ref{tab:ksi:quantiles}. The presented results further validate the suggested value for $\xi$ in Section~\ref{sec:exp} of the main paper. We also remark that the product $\eta\xi$ is not always less than 1 in our experiments. Hence, although we needed the condition $\eta\xi < 1$ to prove our theoretical result in Section~\ref{sec:theory}, the numerical evidence shows that the product can take values greater than 1 and still attain convergence (however $\eta\xi$ should not become arbitrarily large).

\begin{figure}[htp]
    \centering
    \includegraphics[width=0.49\textwidth]{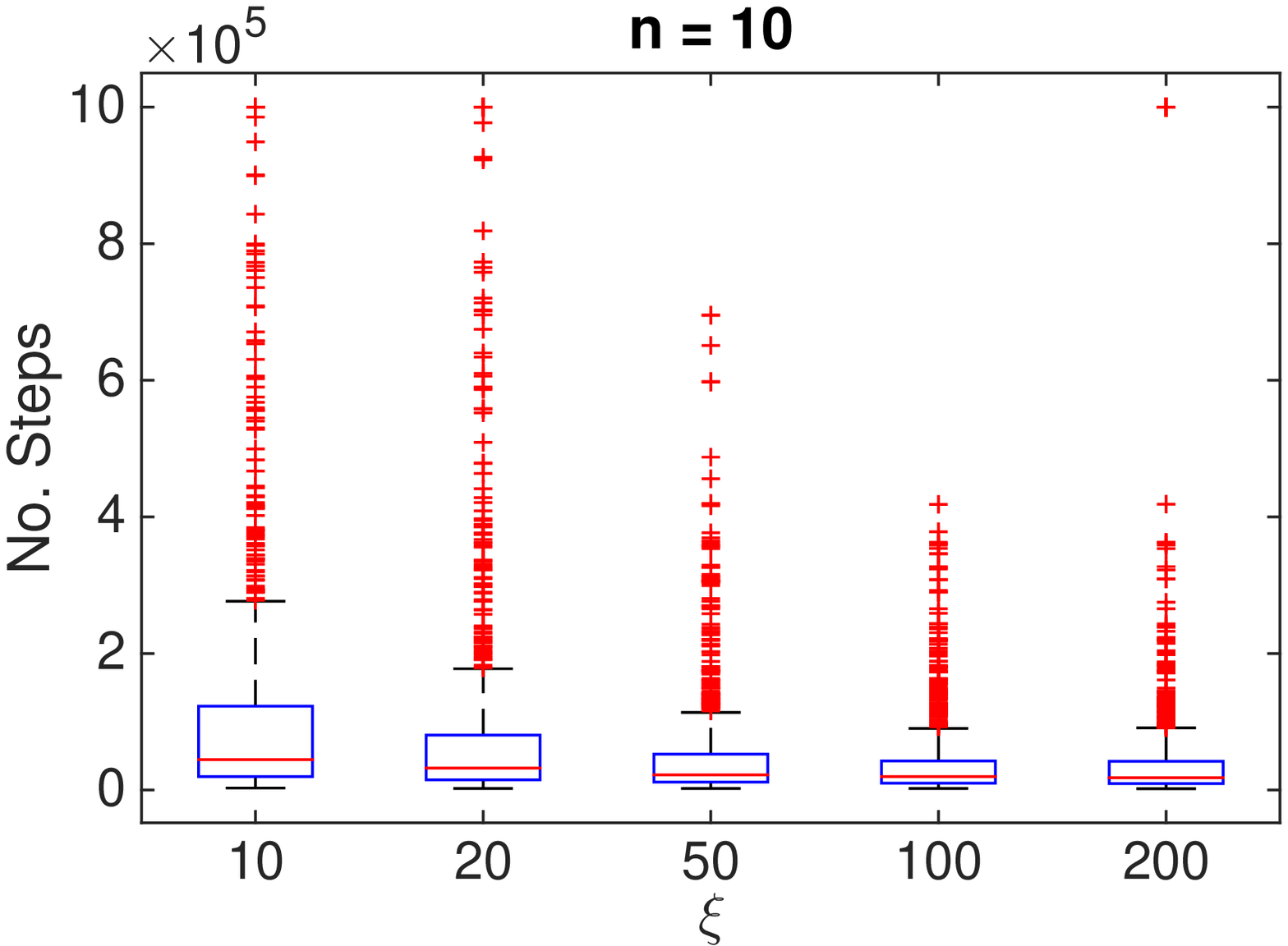}
    \includegraphics[width=0.49\textwidth]{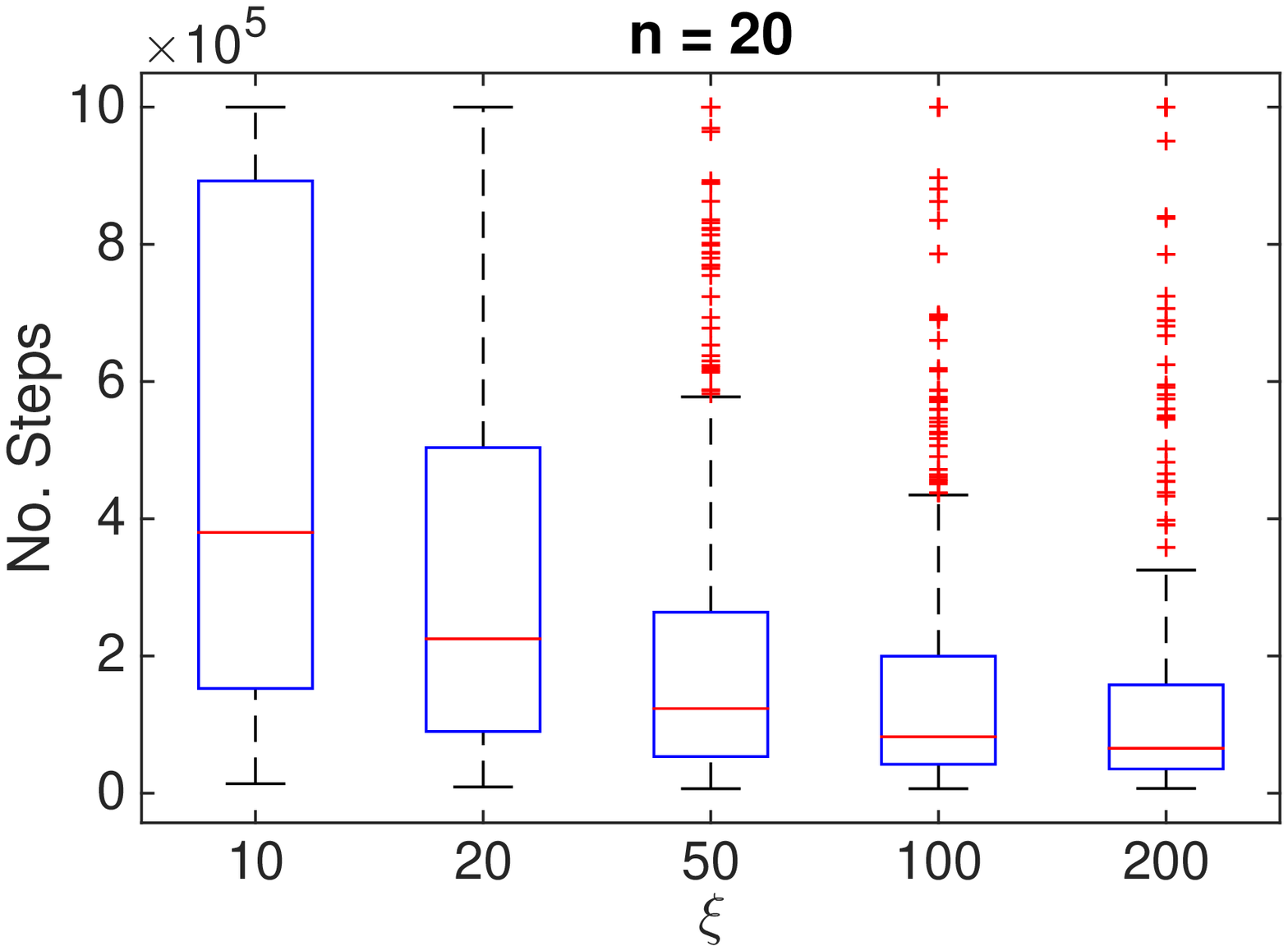}
    \caption{Boxplots for the number of steps until convergence for various values of $\xi$ and two payoff matrix sizes.}
    \label{ksi:effect:fig}
\end{figure}

\begin{table}[th]
\caption{Quantile statistics on the number of steps till convergence for various values of $\xi$ and $n$. The maximum number of steps was set to $t_{\max} = 2\times 10^6$.}
\label{tab:ksi:quantiles}
\centering
\begin{tabular}{ccccccc}
Matrix size & Quantile & $\xi=10$ & $\xi=20$ & $\xi=50$ & $\xi=100$ & $\xi=200$ \\ \hline 
\multirow{3}{*}{$n=10$} & 75\% & 127.1K & 83.3K & 54.5K & 44.1K & 43.6K \\
& 90\% & 346.7K & 209.1K & 137.9K & 111.3K & 110.3K \\
& 97.5\% & 1035.2K & 640.3K & 322.1K & 228.7K & 372.8K \\ \hline 
\multirow{3}{*}{$n=20$} & 75\% & 1957.2K & 1127.2K & 576.9K & 441.3K & 342.7K \\
& 90\% & 2000.0K & 2000.0K & 1644.8K & 1076.0K & 830.8K \\
& 97.5\% & 2000.0K & 2000.0K & 2000.0K & 2000.0K & 2000.0K \\
\end{tabular}
\end{table}

\noindent {\bf Number of steps.} Moving on, we present additional comparisons between FLBR-MWU and OMWU. Figure \ref{n:effect:fig} demonstrates the distribution of the number of steps till convergence for FLBR-MWU  (left) and OMWU (right). Interestingly, the distribution for payoff matrix size $n=50$ with the FLBR-MWU algorithm is similar to the distribution for $n=5$ with the OMWU algorithm. The computational gains are expected to be even more dramatic for larger games.

\begin{figure}[htp]
    \centering
    \includegraphics[width=0.49\textwidth]{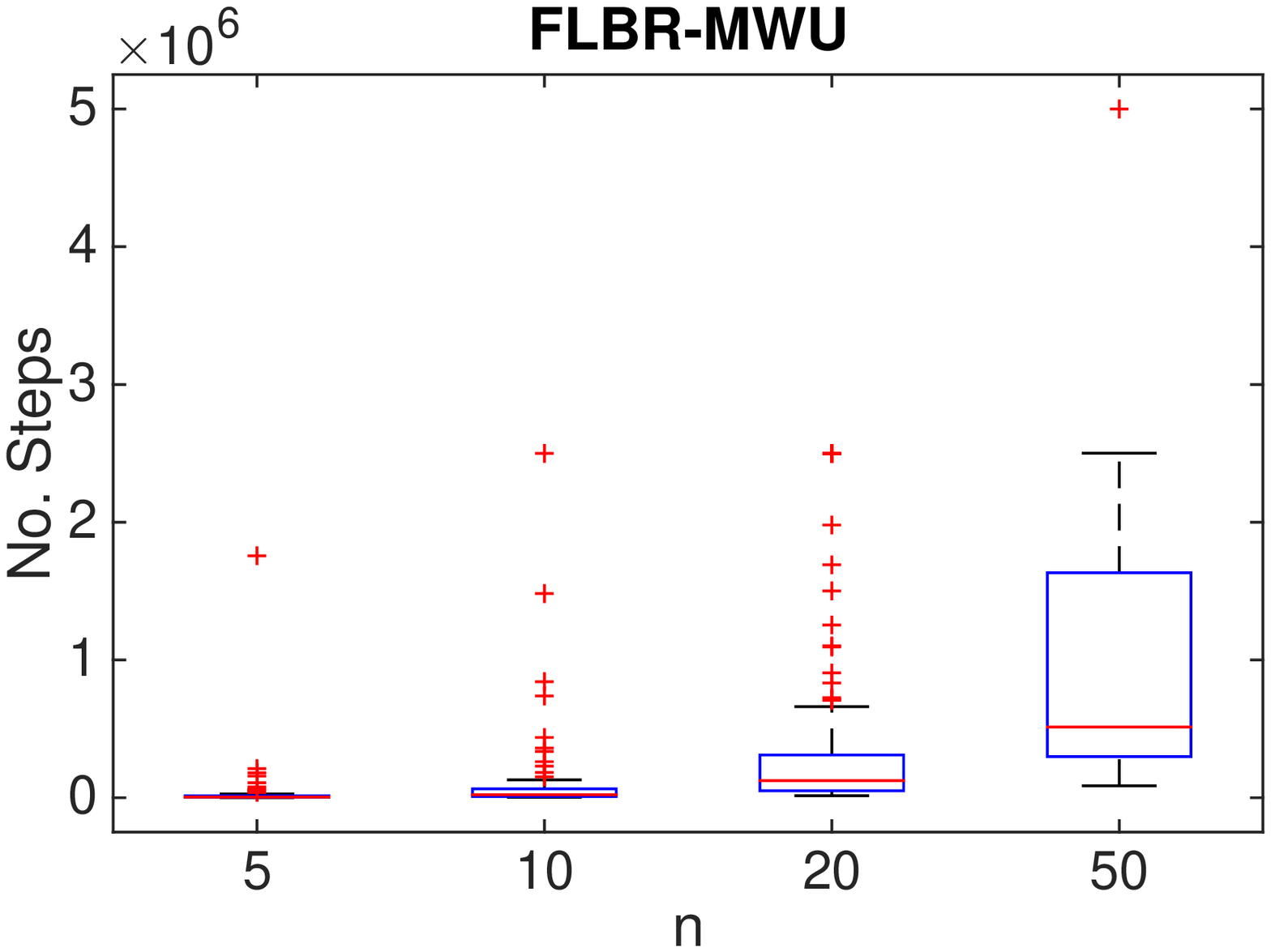}
    \includegraphics[width=0.49\textwidth]{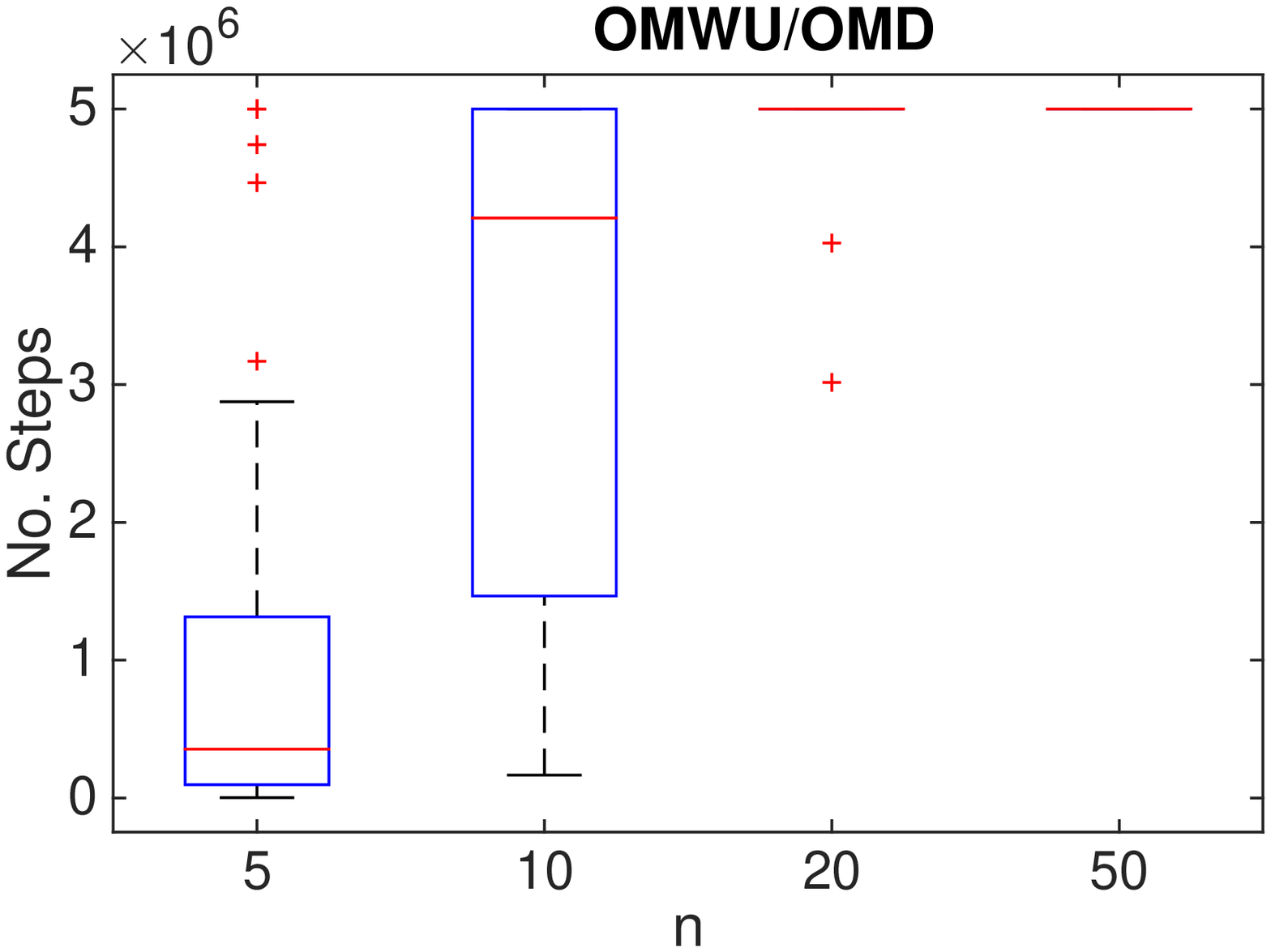}
    \caption{Boxplots for the number of steps until convergence for various payoff matrix sizes under FLBR-MWU (left) and OMWU/OMD (right). The computational gains when FLBR-MWU is used are striking.}
    \label{n:effect:fig}
\end{figure}

\noindent {\bf MWU, OMWU, and OMD.} We also present a comparison among the MWU, OMWU and OMD dynamics (where for OMD we implemented the version of \cite{DBLP:conf/iclr/MertikopoulosLZ19} with entropy regularization). Figure \ref{long:runs:fig} shows the evolution of a long run of 5 million steps and two values for the learning rate, $\eta$. We use the same payoff matrix as in Figure~\ref{motiv:ex:fig} of the main text and recall that the proposed FLBR-MWU method converged after only $100K$ steps (see Figure~\ref{motiv:ex:fig} in the main text). It is evident from the KL divergence in Figure \ref{long:runs:fig} (leftmost panels) that the OMWU and OMD algorithms have almost the same behavior, as expected by \cite{Wei2021LinearLC}, and they both converge, but in a very slow pace. The oscillatory behavior is prominent even after a large number of steps, as quantified by the $l_1$ norm difference (rightmost panels of Figure \ref{long:runs:fig}).

\begin{figure}[htp]
    \centering
    \includegraphics[width=0.49\textwidth]{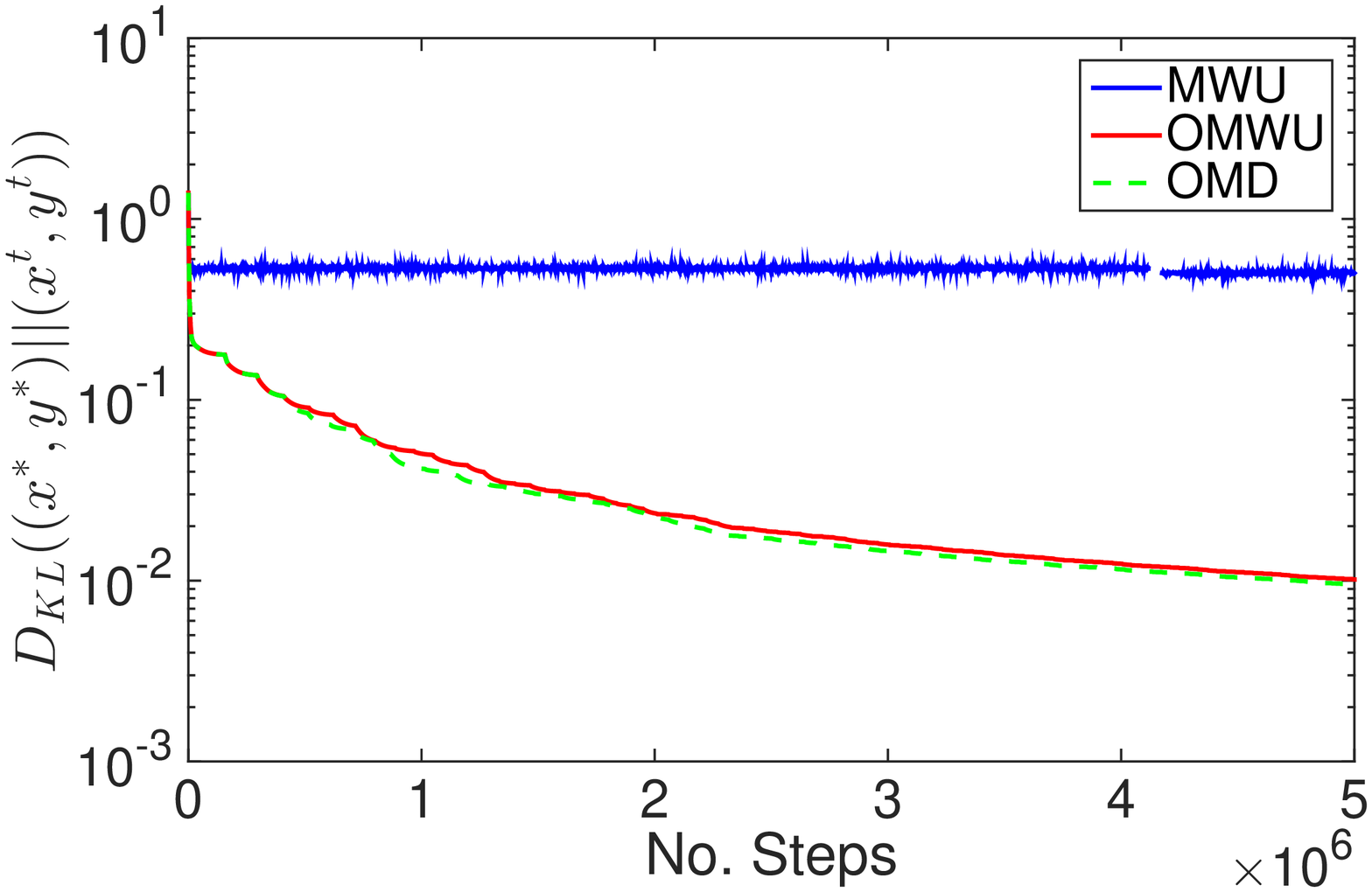}
    \includegraphics[width=0.49\textwidth]{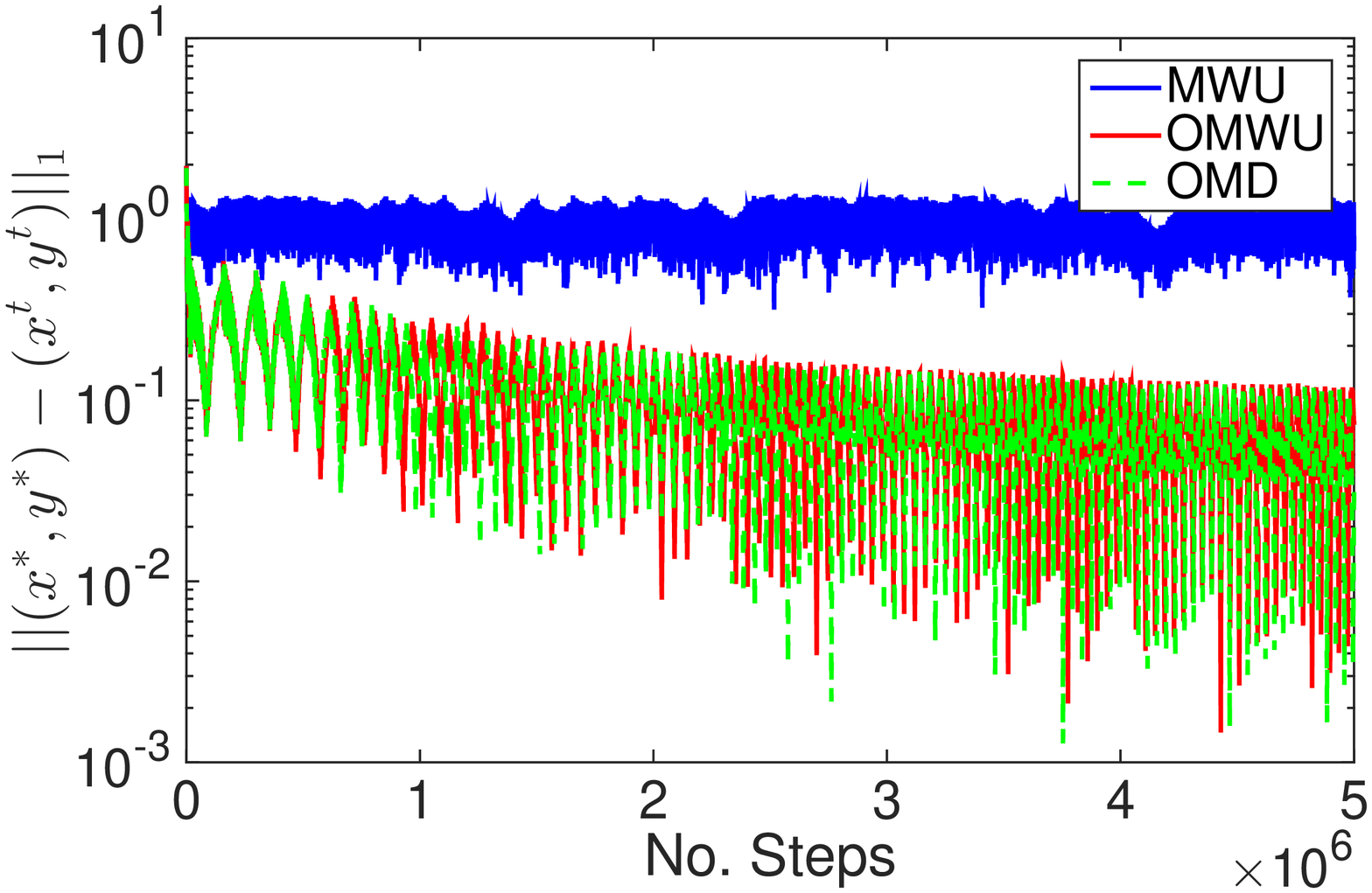}
    \includegraphics[width=0.49\textwidth]{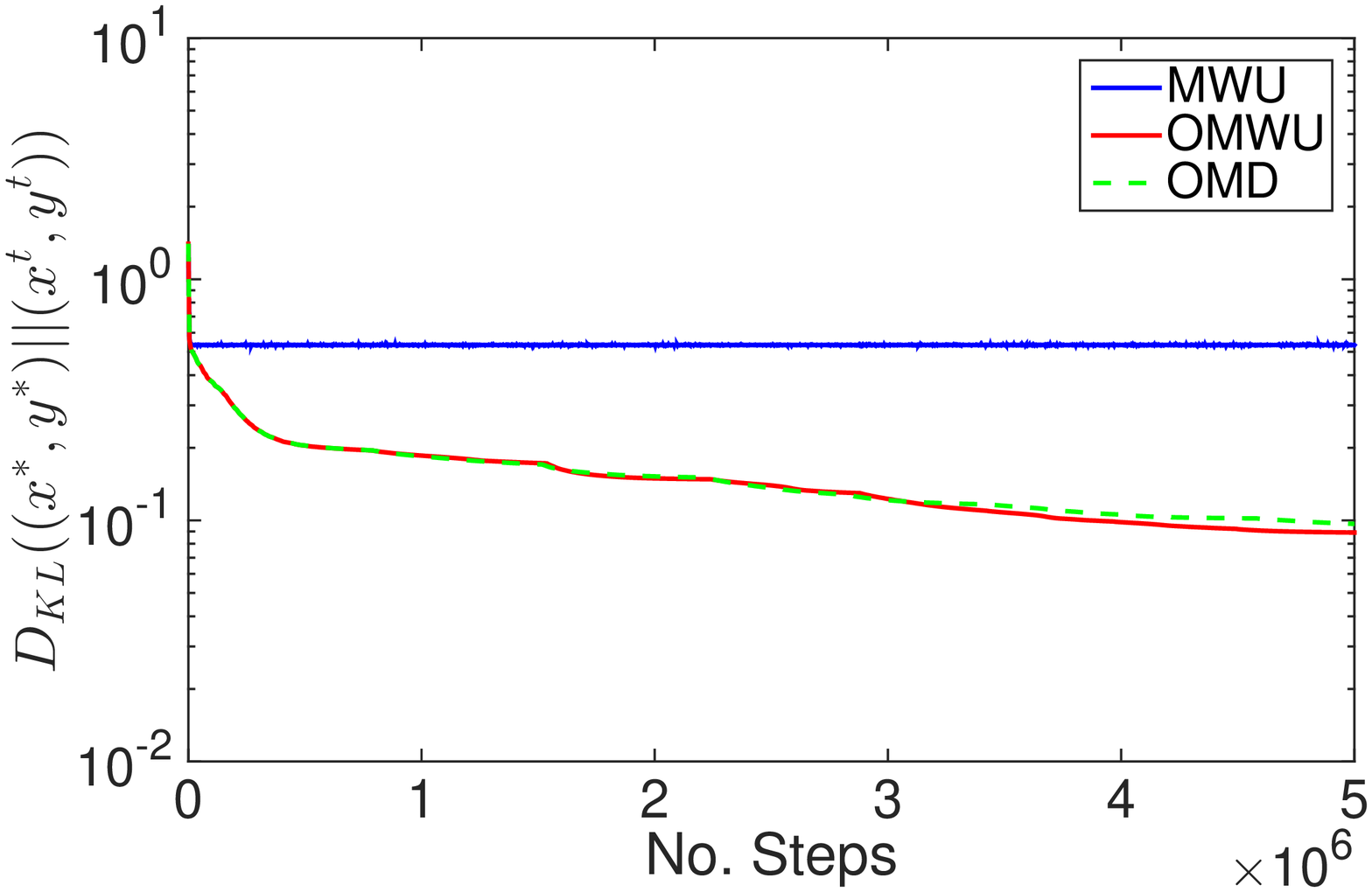}
    \includegraphics[width=0.49\textwidth]{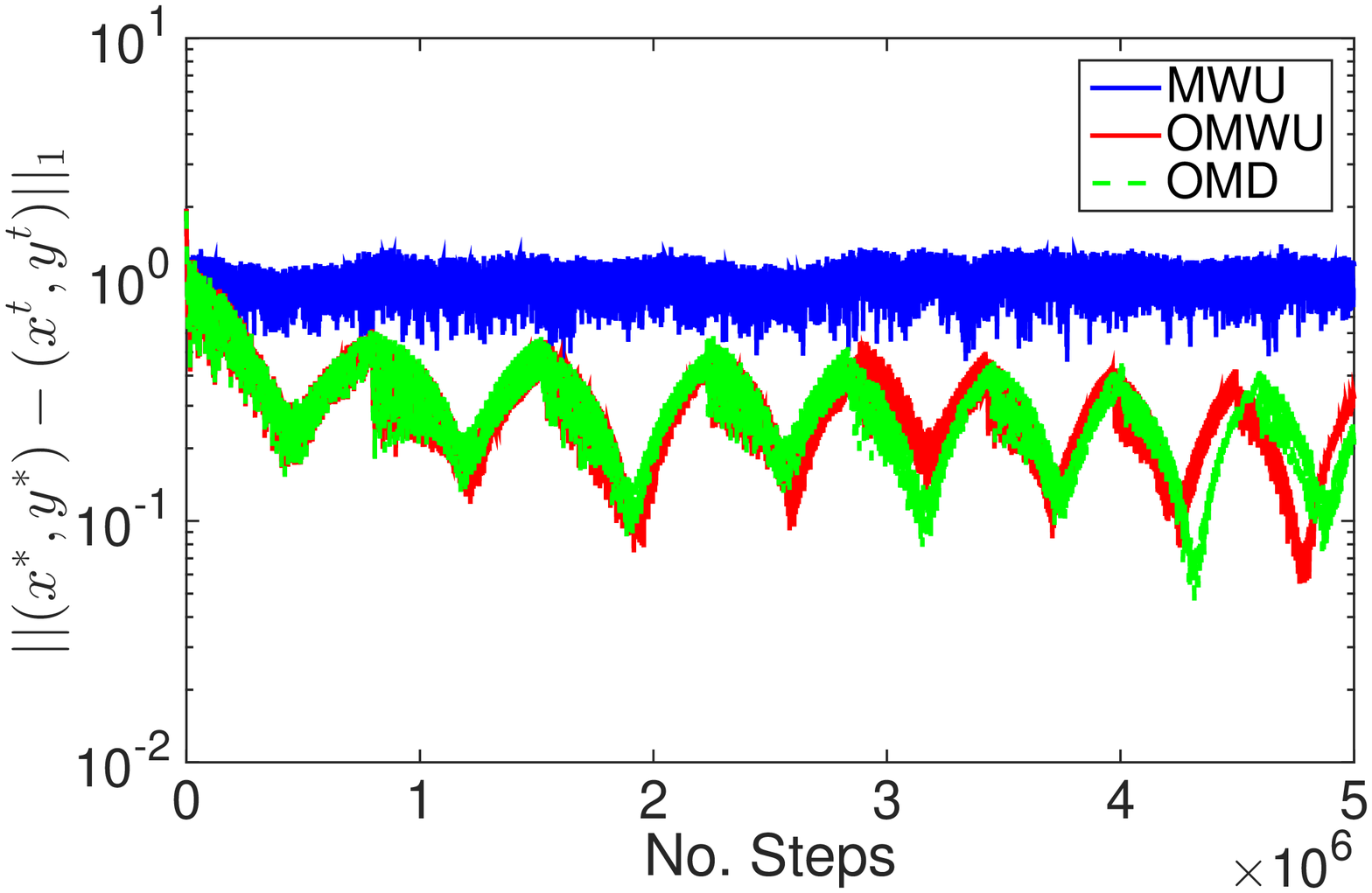}
    \caption{KL divergence and $l_1$ norm difference for $t_{\max}=5\times 10^6$ and two values for the learning rate: $\eta=0.1$ (upper row of panels) and $\eta=0.02$ (lower row of panels).}
    \label{long:runs:fig}
\end{figure}

Finally, we report in Table~\ref{tab:methods} several convergence statistics between OWMU, OMD and FLBR-MWU algorithms with $\eta = 0.1$. This table is an extension of Table~\ref{tab:size} from the main text (in Section~\ref{sec:exp}). Once again, the proposed FLBR-MWU algorithm is orders of magnitude faster while the closeness of the statistics between OWMU and OMD reveals the (almost) equivalence between the two algorithms.

\begin{table}[th]
\caption{Statistics on the number of steps till convergence for OWMU, OMD and FLBR-MWU and various payoff matrix sizes. The maximum number of steps was set to $t_{\max} = 5\times 10^6$.}
\label{tab:methods}
\centering
\begin{tabular}{cccccc}
Matrix size & Statistic & $n=5$ & $n=10$ & $n=20$ & $n=50$ \\ \hline 
\multirow{3}{*}{OWMU} & Mean & 1287.3K & 3280.9K & 4997.8K & 5000.0K \\
& Median & 631.9K & 3697.8K & 5000.0K & 5000.0K \\
& $t_{\max}$ & 12.0 & 44.0 & 98.0 & 100.0  \\ \hline 
\multirow{3}{*}{OMD} & Mean & 1287.6K & 3292.9K & 4997.8K & 5000.0K \\
& Median & 631.9K & 3629.1K & 5000.0K & 5000.0K  \\
& $t_{\max}$ & 12.0 & 44.0 & 98.0 & 100.0 \\ \hline 
\multirow{3}{*}{FLBR-MWU} & Mean & 18.8K & 45.9K & 267.1K & 1130.8K \\
& Median & 8.0K & 21.4K & 64.0K & 701.3K  \\
& $t_{\max}$ & 0.0 & 0.0 & 0.0 & 2.0 \\
\end{tabular}
\end{table}


\end{document}